\documentclass[11pt]{article}

\usepackage{amsmath,amsfonts,amssymb,nccmath}

\usepackage{cancel}
\usepackage{tikz-cd}
\usetikzlibrary{arrows}

\usepackage{float}

\usepackage[numbers,sort&compress]{natbib}
\usepackage{graphicx}

\usepackage{hyperref}
\hypersetup{
   colorlinks   =  true,
    linkcolor    = blue,
    citecolor    = red,
     urlcolor	=magenta,     
}

\usepackage{amsthm}
\theoremstyle{definition} 

\theoremstyle{theorem} 
\newtheorem{theorem}{Theorem}

\newtheorem{proposition}{Proposition}

\newcommand{\sect}[1]{\setcounter{equation}{0}\section{#1}}

\def\be{\begin{equation}}
\def\ee{\end{equation}}
\def\bea{\begin{eqnarray}}
\def\eea{\end{eqnarray}}

\DeclareMathOperator\spn{span}

\newcommand{\kk}{\kappa}

\newcommand{\Sk}{{\rm\ \!S}}            
\newcommand{\Ck}{{\rm\ \!C}}           
 \newcommand{\Tk}{{\rm\ \!T}}

  \newcommand{\otra}{\lambda}

\begin{document}

\
  \vskip0.5cm

\begin{center}
\baselineskip 24 pt {\Large \bf  
{Higher-order superintegrable momentum-dependent Hamiltonians on curved spaces from the classical Zernike system}}
\end{center}

\medskip 
\medskip

\begin{center}

{\sc Alfonso Blasco$^1$, Ivan Gutierrez-Sagredo$^{2}$ and Francisco J.  Herranz$^1$ }

\medskip

{$^1$Departamento de F\'isica, Universidad de Burgos, 
09001 Burgos, Spain}

{$^2$Departamento de Matem\'aticas y Computaci\'on, Universidad de Burgos, 
09001 Burgos, Spain}

 \medskip
 
e-mail: {\href{mailto:ablasco@ubu.es}{ablasco@ubu.es}, \href{mailto:igsagredo@ubu.es}{igsagredo@ubu.es}, \href{mailto:fjherranz@ubu.es}{fjherranz@ubu.es}}

\end{center}

\medskip

\begin{abstract}
\noindent
We consider the classical momentum- or velocity-dependent two-dimensional Hamiltonian given by
$$
\mathcal H_N =  p_1^2 + p_2^2 +\sum_{n=1}^N \gamma_n(q_1 p_1 + q_2 p_2)^n  ,
$$
where $q_i$ and $p_i$ are generic canonical variables, $\gamma_n$ are arbitrary coefficients, and   $N\in \mathbb N$. For $N=2$, being both $\gamma_1,\gamma_2$ different from zero, this reduces to the 
classical Zernike system.  We prove that $\mathcal H_N$  always provides a superintegrable system (for any value of  $\gamma_n$ and $N$)  by obtaining the corresponding constants of the motion explicitly, which turn out to be of  higher-order in the momenta.  Such generic results are not only applied to the Euclidean plane, but also to the sphere and the hyperbolic plane. In the latter curved spaces, $\mathcal H_N $ is   expressed in geodesic polar coordinates showing that such a new superintegrable Hamiltonian   can be regarded as a superposition of   the isotropic 1\,:\,1   curved (Higgs) oscillator  with even-order     anharmonic curved oscillators plus another superposition of  higher-order momentum-dependent potentials.
  Furthermore, the symmetry algebra determined by the 
constants of the motion is also studied, giving rise to a $(2N-1)$th-order polynomial  algebra.  As a byproduct, the   Hamiltonian $\mathcal H_N $ is  interpreted as a family of  superintegrable perturbations of the classical Zernike system. Finally,   it is  shown that $\mathcal H_N$ (and so the Zernike system as well) is endowed with a   Poisson $\mathfrak{sl}(2,\mathbb R)$-coalgebra symmetry which would allow for further possible generalizations that are also  discussed.

\end{abstract}

\medskip
\medskip

 \noindent {MSC}: 37J35, 70H06, 22E60, 17B62

  \medskip

  \noindent {PACS}: 02.30.Ik, 45.20.Jj, 02.20.Sv, 02.40.Ky

  \medskip

    \noindent{Keywords}: Integrable systems; Curvature; Sphere; Hyperbolic plane; Curved oscillator;  Poisson   coalgebras;  Integrable perturbations; symmetry algebras 
 \newpage

\tableofcontents

\sect{Introduction}

Classical and quantum momentum- or velocity-dependent Hamiltonian systems have been extensively studied in the literature over many decades   mainly due to their relevant,  wide and varied physical applications.
Without trying to be exhaustive, let us mention  that linear momentum-dependent Hamiltonians have been considered from different viewpoints in~\cite{Hietarinta84,Hietarinta85,Gunn1985,DGRW1985,IcBo1988,McSW2000,Ranada2003,Pucacco2004,Ranada2005,MB2008,SozTsi2015,Tsi2015,Yehia2016,BKS2020,FSW2020} and  quadratic momentum-dependent ones have been analyzed in~\cite{RFL62,McM1965,FGS1967,SesmaVento1978,SBB2012}. In addition, exponentials of momentum-dependent potentials  ($V\propto {\rm e}^{-\,p^2}$) have also been considered   in~\cite{DDR1987,BG1988,CH1995momentumpotentials,LGXZL2003} and another more involved 
 momentum-dependent potential was recently introduced in~\cite{NMS2020}; see references therein in all the aforementioned works.

Furthermore, from a completely different perspective,  we stress that   quantum groups~\cite{ChariPressley1994}  have been applied to classical and quantum (super)integrable Hamiltonians through both deformed and undeformed coalgebras in~\cite{BR98,BBHMR2009}. Following this coalgebra  symmetry approach,  several classes of  momentum-dependent classical Hamiltonians have been constructed in~\cite{BBHMR2009,Ballesteros2000,Ballesteros2007,Ballesteros2009} giving rise to quasi-maximally superintegrable systems, {\em i.e.}~in arbitrary dimension $d$ they are endowed, by construction, with $(2d-3)$ functionally independent constants of the motion (besides the Hamiltonian). Hence one additional constant of the motion is left to ensure maximal superintegrability.

In this paper, we shall consider a large class of two-dimensional (2D) higher-order momentum-dependent systems comprised within the Hamiltonian given by
\begin{equation}
\mathcal H_N =  p_1^2 + p_2^2 +\sum_{n=1}^N \gamma_n(q_1 p_1 + q_2 p_2)^n ,
\label{00}
\end{equation}
where $q_i$ and $p_i$ are generic canonical variables (with Poisson bracket $\{q_i,p_j\}=\delta_{ij}$), $\gamma_n$ are arbitrary coefficients, and the index $N\in \mathbb N$. Therefore, as particular cases,  we find that for $N=1$ we shall deal with linear momentum-dependent  Hamiltonians and for for $N=2$  with quadratic momentum-dependent ones, but when $N>2$ we shall obtain cubic, quartic\dots momentum-dependent  Hamiltonians.

The underlying motivation to consider $\mathcal H_N$ (\ref{00}) is that this is just the natural generalization (for arbitrary $N$) of the  superintegrable classical Zernike system formerly introduced in~\cite{PWY2017zernike} (see also~\cite{Fordy2018,Wolf2020}), which is  recovered for $N=2$
   with $\gamma_1\ne 0$ and $\gamma_2\ne 0$. Recall that the original Zernike system is  properly quantum~\cite{Zernike1934} and   
    as a quantum superintegrable Hamiltonian has been extensively studied   in~\cite{PSWY2017,PWY2017a,Atakishiyev2017,Fordy2018,Wolf2020,Atakishiyev2019}.     Moreover, we observe that the Hamiltonian $\mathcal H_N$ (\ref{00}) is naturally endowed with  a Poisson $\mathfrak{sl}(2,\mathbb R)$-coalgebra symmetry~\cite{Ballesteros2007,BBHMR2009}.

The aim of this paper is twofold. On the  one hand, we   explicitly prove that the   Hamiltonian $\mathcal H_N$ (\ref{00}) is superintegrable for any $N$ and for any value of the coefficients $\gamma_n$. And, on the other hand, we   apply this result not only to the   flat Euclidean plane $\mathbf E^2$, but also to the curved sphere $\mathbf S^2$ and the hyperbolic plane $\mathbf H^2$.

The structure of the paper    is as follows. In the next section we review the    classical Zernike Hamiltonian on  $\mathbf E^2$ along with its interpretation on   $\mathbf S^2$ and $\mathbf H^2$ and, furthermore, we describe its underlying Poisson $\mathfrak{sl}(2,\mathbb R)$-coalgebra symmetry.
This allows us to propose  $\mathcal H_N$ (\ref{00})  as its natural generalized Hamiltonian.  In Section~\ref{s31} we prove that   $\mathcal H_N$  always determines a superintegrable system on $\mathbf E^2$  (for any   $\gamma_n$ and $N$) by obtaining explicitly the   constants of the motion, which turn out to be 
of higher-order in the momenta. The corresponding interpretation on  $\mathbf S^2$ and $\mathbf H^2$ is   performed in Section~\ref{s32}. In particular, we introduce the so called geodesic polar coordinates~\cite{RS,conf,BaHeMu13}, which are the curved generalization of the usual Euclidean polar coordinates. In this way, we show that $\mathcal H_N$ can alternatively be regarded as a superposition of   the isotropic 1\,:\,1   curved (Higgs) oscillator  with even-order     anharmonic curved oscillators plus another superposition of  higher-order momentum-dependent potentials.

Such general results are illustrated in Section~\ref{s4} for $N\le 8$ and, moreover, the associated polynomial symmetry algebra, defined through the constants of the motion, is also computed leading to a  $(2N-1)$th-order generalization of    the   well-known cubic  Higgs Poisson algebra~\cite{PWY2017zernike,Higgs}. As a byproduct, our results are specifically applied to the classical Zernike Hamiltonian in Section~\ref{s5}, being interpreted  as superintegrable perturbations. Their real   part of  the  trajectories are also plotted up to $N=6$.

We remark that the underlying Poisson $\mathfrak{sl}(2,\mathbb R)$-coalgebra symmetry of  $\mathcal H_N$  
naturally suggests further   possible generalizations. These open problems    along with  the application to (1+1)D Lorentzian spacetimes of constant curvature (Minkowskian and (anti-)de Sitter spaces) are    discussed in the last section with some detail. 
 To end with, we stress that a quantization of $\mathcal H_N$   is also addressed in the last section. The guiding idea is to replace the  
Poisson $\mathfrak{sl}(2,\mathbb R)$-coalgebra symmetry by a  Lie  $\mathfrak{gl}(2)$-coalgebra symmetry. Anyhow,    serious ordering problems arise in the constants of the motion, so that  our proposal for a  quantum $\mathcal H_N$   Hamiltonian also remains as an open problem.

\newpage


\sect{The classical Zernike system revisited}
\label{s2}

The original quantum  Zernike system was     introduced in~\cite{Zernike1934} and, very recently, deeply analysed  in~\cite{PSWY2017,PWY2017a,Atakishiyev2017,Fordy2018,Wolf2020,Atakishiyev2019}  (see also references therein) as a quantum superintegable Hamiltonian. Such a Hamiltonian system is   defined on the 2D  Euclidean plane $\mathbf E^2$, which has   a potential depending on  both linear and quadratic terms on the quantum momenta operators. Its classical  counterpart 
was formerly presented and studied in~\cite{PWY2017zernike} (see also~\cite{Fordy2018,Wolf2020}), which possesses quadratic in the momenta constants of the motion.

The superintegrable classical Zernike system is the cornerstone of our construction of new higher-order superintegrable momentum-dependent classical Hamiltonians on a 2D Riemannian space of constant (Gaussian) curvature $\kappa$, so covering the flat Euclidean space $\mathbf E^2$ ($\kappa=0$),  the sphere  $\mathbf S^2$ ($\kappa>0$) and the hyperbolic or Lobachevski space  $\mathbf H^2$  ($\kappa<0$). With this aim,  we review  in this section  the main results on the 
known classical Zernike system along with its interpretation on curved spaces and, furthermore, we present   new properties related with Poisson $\mathfrak{sl}(2,\mathbb R)$-coalgebra symmetry~\cite{BR98,Ballesteros2007,BBHMR2009,Latini2019,Latini2021} which, to the best of our knowledge, have not been considered in the literature yet.

The main superintegrability properties (in the Liouville sense~\cite{Perelomov}) of the classical Zernike system are established in the following statement.

 \begin{theorem}
\cite{PWY2017zernike}
\label{teor0}
Let $\{ q_1,q_2,p_1,p_2\}$ be a set of canonical variables with Poisson brackets $\{ q_i,p_j\} = \delta_{ij}$. The classical Zernike  Hamiltonian   on the Euclidean plane, $(q_1,q_2)\equiv (x,y)\in \mathbb R^2$,  is given by 
\begin{equation}
\mathcal H_{\rm Zk} =  p_1^2 + p_2^2 + \gamma_1(q_1 p_1 + q_2 p_2)+ \gamma_2(q_1 p_1 + q_2 p_2)^2 ,
\label{za}
\end{equation}
where $\gamma_1$ and  $\gamma_2$   are arbitrary parameters.

\noindent
(i) The Hamiltonian $\mathcal H_{\rm Zk} $ has three (quadratic in the momenta) constants of the motion:
\bea
&& \mathcal C = q_1 p_2 - q_2 p_1 , \nonumber\\[2pt]
&&  \mathcal I=  p_2^2 + \gamma_1\,  q_2p_2   + \gamma_2  \bigl( q_1^2 + q_2^2 \bigr)p_2^2 , \label{zb}\\
&&  \mathcal I'=  p_1^2 + \gamma_1\, q_1 p_1   + \gamma_2 \bigl( q_1^2 + q_2^2 \bigr)p_1^2 .
\nonumber
\eea
(ii) The above functions fulfil the relation
\be
\mathcal \mathcal H_{\rm Zk} = \mathcal I + \mathcal I' - \gamma_{2}\,\mathcal C^{2}  .
\label{zc}
\ee
(iii)  The sets  $\{\mathcal H_{\rm Zk} , \mathcal I, \mathcal C\}$ and   $\{\mathcal H_{\rm Zk} , \mathcal I', \mathcal C\}$   are    formed by three functionally independent functions so that $\mathcal H_{\rm Zk}$ is a superintegrable Hamiltonian.

\noindent
(iv) The three functions defined by
\bea
&&\mathcal L_1:= \mathcal C/2, \qquad  \mathcal L_2:=\bigl(\mathcal I'  -  \mathcal I \bigr)/2 ,\nonumber\\ 
&&\mathcal L_3:= \{  \mathcal L_1,\mathcal L_2\}= \left( 1+ \gamma_2 \bigl(q_1^2+q_2^2\bigr)  \right)p_1 p_2 + \tfrac 12 \gamma_1 (q_1 p_2 + q_2 p_1) ,
\label{zzd}
\eea
satisfy the Poisson brackets
\be
\{ \mathcal L_1,\mathcal  L_2\}=\mathcal L_3, \qquad \{ \mathcal L_1,\mathcal  L_3\}=-\mathcal L_2, \qquad  \{ \mathcal L_2, \mathcal L_3\}=-\mathcal L_1\left( \gamma_1^2 + 2 \gamma_2 \mathcal H_{\rm Zk}+8  \gamma_2^2  \mathcal L_1^2  \right) .
\label{zd}
\ee
\end{theorem}

All   the  results covered by Theorem~\ref{teor0} can be  expressed straightforwardly  in polar coordinates $(r,\phi)$ and conjugate momenta $(p_r,p_\phi)$, as it was already performed in~\cite{Fordy2018,PWY2017zernike},  by means of the usual canonical transformation given by
\be
\begin{array}{ll}
q_1=r\cos\phi ,   &\quad\displaystyle{   p_1=\cos\phi \, p_r -\frac{\sin \phi}{r}\, p_\phi     } , 
  \\[4pt]
q_2=r\sin\phi ,   &\quad\displaystyle{   p_2=\sin\phi \, p_r + \frac{\cos \phi}{r}\, p_\phi    } .
\end{array}
\label{ze}
\ee
In particular, in these variables the Hamiltonian $\mathcal H_{\rm Zk} $ (\ref{za}) and the angular momentum constant of the motion $ \mathcal C $ (\ref{zb}) 
turn out to be
\be
\mathcal H_{\rm Zk} =  p_r^2 + \frac{p_\phi^2 }{r^2} + \gamma_1 r p_r+ \gamma_2(r p_r)^2 , \qquad   \mathcal C = p_\phi ,
\label{zf}
\ee
showing directly the integrability of the system,  while   $ \mathcal I$ (or $ \mathcal I'$)  (\ref{zb})  is an additional integral (or hidden symmetry) determining its superintegrability.


\subsection{Interpretation on the   sphere and the hyperbolic space}
\label{s21}

We stress, as it  was already pointed  out in~\cite{PWY2017zernike}, that the relations (\ref{zd}) provide a cubic Higgs algebra~\cite{Higgs} (whenever   $\gamma_2\ne 0$), which is just the symmetry algebra of the integrals of the motion of the  well-known Higgs or isotropic curved oscillator on the 2D sphere $\mathbf S^2$ that has been extensively studied over the last few decades~\cite{Higgs,Leemon,Pogoa,RS,Kalnins1,Kalnins2,Nersessian1,Ranran,Santander6,BaHeMu13,MiPWJPa13,GonKas14AnnPhys,BaBlHeMu14,Kuruannals} (see also references therein). We  also  recall that a cubic  Higgs-type algebra   arises in   Kepler--Coulomb systems 
on   $\mathbf S^2$ and on the 2D hyperbolic space $\mathbf H^2$~\cite{Kepler2009}. These facts suggest a natural relationship between the previous interpretation of $\mathcal H_{\rm Zk} $ on   $\mathbf E^2$ and an alternative one as a  superintegrable Hamiltonian on a 2D curved space as it was   mentioned in~\cite{Fordy2018,PWY2017zernike}.

Let us consider the terms in $\mathcal H_{\rm Zk}$ depending quadratically in the momenta as the free Hamiltonian or kinetic energy of the system,  so that the associated metric can then be deduced.  From the expression (\ref{zf}) in polar variables  the underlying 2D non-Euclidean metric reads
\be
{\rm d}s^2= \frac{1}{1+\gamma_2 r^2}\, {\rm d}r^2+r^2 {\rm d}\phi^2 .
\label{zg}
\ee
Its Gaussian curvature $\kappa$ turns out to be constant and equal to  $-\gamma_2$~\cite{Fordy2018}. Hence, according to the sign of the curvature  parameter  $\kappa=-\gamma_2$, we find that the metric (\ref{zg}) simultaneously comprises   the flat Euclidean space $\mathbf E^2$ ($\kappa=\gamma_2=0$),  the sphere  $\mathbf S^2$ ($\kappa>0, \gamma_2<0$) and the hyperbolic space  $\mathbf H^2$  ($\kappa<0, \gamma_2>0$). 
 Since both initial (arbitrary) $\gamma_1$- and  $\gamma_2$-potentials are essential to deal with the proper Zernike system we shall assume  in this section that they are   different from zero, so that we shall deal with $\mathbf S^2$ and $\mathbf H^2$.

It should be noted that the polar radial coordinate  $r$ is no longer a geodesic distance in a curved space
with $\kappa\ne 0$ (which is our case now).  In order to perform an appropriate geometrical and dynamical interpretation of $\mathcal H_{\rm Zk} $ on $\mathbf S^2$ and $\mathbf H^2$, let us introduce the so-called geodesic radial coordinate~\cite{RS,BaHeMu13,conf}, here denoted by $\rho$, which is just the   distance along the geodesic joining the origin in the curved space and the particle,  keeping  unchanged the usual angular coordinate $\phi$. The relationship between $r$ and $\rho$ 
is given by  
\be
r= \Sk_\kk( \rho),\qquad \kk=-\gamma_2,
 \label{zh} 
\ee
where from now on we shall make use of the curvature-dependent cosine and sine functions defined by~\cite{RS,trigo,conf}
\begin{equation}
\Ck_{\kk}(x):=\left\{
\begin{array}{ll}
  \cos{\sqrt{\kk}\, x} &\quad  \kk>0 \\ 
\qquad 1  &\quad
  \kk=0 \\
\cosh{\sqrt{-\kk}\, x} &\quad   \kk<0 
\end{array}\right.   ,
 \qquad 
    \Sk{_\kk}(x) :=  \left\{
\begin{array}{ll}
  \frac{1}{\sqrt{\kk}} \sin{\sqrt{\kk}\, x} &\quad  \kk>0 \\ 
\qquad x  &\quad
  \kk=0 \\ 
\frac{1}{\sqrt{-\kk}} \sinh{\sqrt{-\kk}\, x} &\quad  \kk<0 
\end{array}\right.  .
\label{zi}
\end{equation}
The
$\kk$-tangent    is defined as
\be
\Tk_\kk(x)  := \frac{\Sk_\kk(x)}  { \Ck_\kk(x)} \, .
\label{zj}
\ee
These  $\kk$-dependent trigonometric functions coincide with  the  circular
and hyperbolic  ones for   $\kk=\pm 1$, while under the
contraction    (or flat limit)   $\kk=0$ they reduce    to the parabolic 
functions:  $\Ck_{0}(x)=1$ and 
$\Sk_{0}(x)=\Tk_{0}(x)=x$.   Under the change of variable (\ref{zh}),   the metric (\ref{zg}) is transformed in its usual form in geodesic polar coordinates $(\rho, \phi)$~\cite{RS,conf}: 
\be
{\rm d}s^2=  {\rm d}\rho ^2+ { \Sk^2_\kk( \rho)}\, {\rm d}\phi^2 .
\label{zk}
\ee
Note that its flat limit $\kk=\gamma_2=0$ leads to the usual metric on  $\mathbf E^2$ in polar coordinates, 
${\rm d}s^2=  {\rm d}r ^2+  r^2\, {\rm d}\phi^2 $, since  $\rho\equiv r$.
 
By taking into account the results presented in~\cite{Fordy2018,PWY2017zernike} in canonical polar    variables $\{r,\phi,p_r,p_\phi\}$ (\ref{ze})   together with the   relation (\ref{zh}), 
we  can  apply the results of  Theorem~\ref{teor0}   for $\mathcal H_{\rm Zk} $ on  $\mathbf E^2$  to  $\mathbf S^2$ and $\mathbf H^2$ in geodesic   polar variables. These are summarized as follows.


\begin{proposition} 
\label{prop0}
Let $\{ \rho,\phi, p_\rho ,p_\phi\}$ be a set of canonical  geodesic polar  variables with Poisson brackets $\{ q_\alpha ,p_\beta \} = \delta_{\alpha\beta}$ where $\alpha,\beta \in\{\rho, \phi\}$.  \\
(i) The classical superintegrable Zernike  Hamiltonian  (\ref{za}) can be expressed in these variables on $\mathbf S^2$ and $\mathbf H^2$, with $\kappa=-\gamma_2$, by applying the 
 canonical transformation given by
\be
\begin{array}{ll}
q_1=\Sk_\kk( \rho)\cos\phi ,   &\quad\displaystyle{   p_1=\cos\phi \, \frac{p_\rho} {\Ck_\kk(\rho) } -\frac{\sin \phi}{\Sk_\kk( \rho) }\, p_\phi     } , 
  \\[10pt]
q_2=\Sk_\kk( \rho) \sin\phi ,   &\quad\displaystyle{   p_2=\sin\phi  \, \frac{p_\rho} {\Ck_\kk(\rho) } +\frac{\cos \phi}{\Sk_\kk( \rho) }\, p_\phi     },
\end{array}
\label{zl}
\ee
leading to
\be
\mathcal H_{\rm Zk} =  p_\rho^2 + \frac{p_\phi^2 }{ \Sk^2_\kk( \rho)} + \gamma_1  \Tk_\kk( \rho)\,p_\rho .
\label{zm}
\ee
The domain for  the variables $(\rho, \phi) $ of $ \mathcal H_{\rm Zk}$  (\ref{zm})   is  given by  $\phi\in[ 0, 2\pi)$ and
\be
\mathbf S^2\ (\kk>0)\!: \ 0< \rho< \frac{ \pi }{2\sqrt{\kk} } \, ,\qquad \mathbf H^2\ (\kk<0)\!: \ 0< \rho< \infty .
\label{zn}
\ee
(ii) The following canonical transformation  
\be
\begin{array}{ll}
q_1=\Sk_\kk( \rho)\cos\phi ,   &\quad\displaystyle{   p_1=  \frac{\cos\phi} {\Ck_\kk(\rho) }\left(  p_\rho- \frac{\gamma_1}{2} \Tk_\kk( \rho) \right) -\frac{\sin \phi}{\Sk_\kk( \rho) }\, p_\phi     } , 
  \\[10pt]
q_2=\Sk_\kk( \rho) \sin\phi ,   &\quad\displaystyle{   p_2=    \frac{ \sin\phi} {\Ck_\kk(\rho)  }\left(  p_\rho- \frac{\gamma_1}{2} \Tk_\kk( \rho) \right)+\frac{\cos \phi}{\Sk_\kk( \rho) }\, p_\phi     },
\end{array}
\label{zl2}
\ee
gives rise to  the  Zernike  system  (\ref{za})   written as a natural Hamiltonian 
\be
\mathcal H_{\rm Zk} =\mathcal T_\kk+\mathcal U_\kk(\rho) ,\qquad \mathcal T_\kk =  p_\rho^2 + \frac{p_\phi^2 }{ \Sk^2_\kk( \rho)}  , \qquad \mathcal U_\kk(\rho)= -\frac{\gamma_1^2}{4}  \Tk^2_\kk( \rho)  ,
\label{zo}
\ee
where $\mathcal T_\kk$ is the kinetic energy on the curved space and $\mathcal U_\kk(\rho)$ is a central   potential. The latter   is just a central or Higgs oscillator,  with centre at the origin on the curved space, whenever the parameter $\gamma_1$ is a pure imaginary number.
\end{proposition}

\begin{proof}
It is immediate from Theorem~\ref{teor0}, the canonical transformations \eqref{zl} and  (\ref{zl2}) together with  the definitions \eqref{zi} and \eqref{zj}.
\end{proof}


Observe that the relation between the two canonical transformations (\ref{zl}) and (\ref{zl2})   simply corresponds to the substitution in (\ref{zm}) given by
\be
 p_\rho =  \tilde p_\rho- \frac{\gamma_1}{2} \Tk_\kk( \rho),
 \label{zzo}
 \ee
and next dropping the tilde in $\tilde p_\rho$ while  keeping $\rho$ as the common conjugate coordinate~\cite{Fordy2018} so obtaining  (\ref{zo}).

 The isometries of the metric (\ref{zk}) associated with the free Hamiltonian  $ \mathcal T_\kk$ (\ref{zo})  turn out to be~\cite{BaHeMu13}
 \be
J_{01}=  \cos\phi\,p_\rho-\frac{\sin\phi}{ \Tk_\kk(\rho) }\, p_\phi,\qquad 
J_{02}= \sin\phi\,p_\rho+\frac{\cos\phi}{\Tk_\kk(\rho)}\, p_\phi ,\qquad J_{12}=p_\phi .
    \label{zp}
  \ee
These functions fulfil the Poisson brackets given by
 \be
  \{J_{12},J_{01}\}=J_{02},\qquad \{J_{12},J_{02}\}=-J_{01},\qquad \{J_{01},J_{02}\}=\kk J_{12}  ,
   \label{zq}
 \ee
thus closing  a Poisson--Lie   algebra isomorphic either to $\mathfrak{so}(3)$ for $\kk>0$ or to $\mathfrak{so}(2,1)$ for $\kk<0$ in agreement with~\cite{Fordy2018}.
Note also that the kinetic term $\mathcal T_\kk$ (\ref{zo}) is just the Casimir of the  Poisson--Lie   algebra (\ref{zq}):
 \be
 \mathcal T_\kk=J_{01}^2+J_{02}^2+\kk J_{12}^2 .
 \label{zr}
 \ee

As we have mentioned in Proposition~\ref{prop0}, the superintegrable potential $\mathcal U_\kk(\rho)$ (\ref{zo}) corresponds to the isotropic 1\,:\,1   or Higgs oscillator on $\mathbf S^2$ and $\mathbf H^2$ when $\gamma_1$ is purely imaginary.  If we set $\gamma_1=2 {\rm i} \omega$ with real parameter $\omega$, then  $\mathcal U_\kk(\rho)=\omega^2\Tk^2_\kk( \rho) $ with  $\omega$ behaving as the frequency of the curved oscillator, that is, 
$\mathcal U_{+1}(\rho)=\omega^2\tan^2 \rho  $ and $\mathcal U_{-1}(\rho)=\omega^2\tanh^2 \rho  $. In this case, the  Zernike  system   has bounded trajectories which are all periodic   and   given by  ellipses in accordance with~\cite{PWY2017zernike}. 
For some trajectories of the Higgs oscillator on $\mathbf S^2$ and $\mathbf H^2$ see also~\cite{BaHeMu13,Kuruannals}.

From the results of Theorem~\ref{teor0} and Proposition~\ref{prop0} it is straightforward to express the Zernike Hamiltonian together with its associated superintegrability properties in other relevant sets of canonical variables such as geodesic parallel and projective (Beltrami and Poincar\'e) ones~\cite{RS,Santander6,BaHeMu13,BaBlHeMu14,Kuruannals}.


\subsection{Poisson $\mathfrak{sl}(2,\mathbb R)$-coalgebra symmetry}
\label{s22}

Let us consider the algebra $\mathfrak{sl}(2,\mathbb R)=\spn\{J_3,J_+,J_- \}$ expressed as a Poisson--Lie   algebra with defining Poisson brackets and 
Casimir  $ {C}$ given by
\begin{equation}
 \{J_3,J_+\}=2 J_+    , \qquad  
\{J_3,J_-\}=-2 J_- ,\qquad   
\{J_-,J_+\}=4 J_3     ,
\label{zs}
\end{equation}
\begin{equation} 
 {C}=  J_- J_+ -J_3^2  . 
\label{zt}
\end{equation} 
Then, as  any Poisson--Lie   algebra,   $\mathfrak{sl}(2,\mathbb R)$ can be endowed with a Poisson coalgebra structure~\cite{BR98},  $(\mathfrak{sl}(2,\mathbb R),\Delta)$, by considering the primitive or non-deformed coproduct map  $\Delta$ given by
\be
\Delta: \mathfrak{sl}(2,\mathbb R)\to \mathfrak{sl}(2,\mathbb R) \otimes \mathfrak{sl}(2,\mathbb R) ,\qquad \Delta(J_l)=  J_l \otimes 1+ 1\otimes J_l  ,\quad l\in\{3,+,-\},
\label{zu}
\ee
which is a homomorphism of Poisson algebras from $(\mathfrak{sl}(2,\mathbb R), \{ \cdot , \cdot \})$ and $(\mathfrak{sl}(2,\mathbb R) \otimes \mathfrak{sl}(2,\mathbb R), \{ \cdot, \cdot \}^{(2)})$, where $\{ \cdot , \cdot \}$ is given by \eqref{zs} and $\{ \cdot, \cdot \}^{(2)}$ is the direct product of two such Poisson structures. Notice that the (trivial) counit and antipode can also be defined giving rise to a non-deformed Hopf algebra structure~\cite{ChariPressley1994}.

A  one-particle symplectic realization of (\ref{zs}) reads 
\begin{equation}
 J_-^{(1)}=q_1^2 ,    \qquad    
J_+^{(1)}=  p_1^2+\frac {\otra_1}{ q_1^2}  ,   \qquad 
J_3^{(1)}= q_1 p_1     ,
\label{zv1}
\end{equation}
where $\otra_1$ is a real parameter that labels the representation through the Casimir (\ref{zt}):  
\be
 {C}^{(1)}=  J_-^{(1)} J_+^{(1)} -\left(J_3^{(1)}\right)^2   =  \otra_1.
\ee 
From (\ref{zv1}),  the coproduct (\ref{zu}) provides the following two-particle symplectic realization of (\ref{zs}):
\be
 J_-^{(2)}=q_1^2+  q_2^2,    \qquad    
J_+^{(2)}=  p_1^2+\frac {\otra_1}{ q_1^2} +p_2^2+\frac{\otra_2}{ q_2^2}  ,   \qquad 
J_3^{(2)}= q_1 p_1    + q_2 p_2   .
\label{zv}
\ee
And the two-particle realization of the Casimir (\ref{zt}) turns out to be:
\be
 {C}^{(2)}=  J_-^{(2)} J_+^{(2)} -\left(J_3^{(2)}\right)^2   = ({q_1}{p_2} -
{q_2}{p_1})^2 + \left(
\otra_1\frac{q_2^2}{q_1^2}+\otra_2\frac{q_1^2}{q_2^2}\right)+ \otra_1+\otra_2 .
\label{zw}
\ee 
By construction~\cite{BR98}, $ {C}^{(2)}$ Poisson-commutes with the three functions  (\ref{zv}) so that any smooth function $ \mathcal H$ defined on them becomes, at least,  a 2D {\em integrable} Hamiltonian, 
\be
\mathcal H^{(2)} = \mathcal H \left(J_3^{(2)}, J_+^{(2)}, J_-^{(2)} \right),
\label{zx}
\ee 
always sharing the constant of the motion given by $ {C}^{(2)}$. Geometrically, the 3D Poisson manifold is foliated by 2D symplectic leaves defined by the level sets of $ {C}^{(2)}$. We recall that, by taking into account the coassociativity property of the coproduct, this result from 2D Poisson $\mathfrak{sl}(2,\mathbb R)$-coalgebra symmetry    can be generalized to arbitrary dimension $d$ providing $(2d-3)$ functionally independent `universal' constants of the motion~\cite{Ballesteros2007,BBHMR2009}; for the corresponding Racah algebra we refer to~\cite{Latini2019,Latini2021}  and references therein. Hence such Hamiltonians are called quasi-maximally superintegrable since only {\em one} additional constant of the motion is left to ensure maximal superintegrability.

The application of the above results to the  classical Zernike  system   is now straightforward. Let us set the parameters $\lambda_1=\lambda_2=0$.
Then the Zernike  Hamiltonian (\ref{za}) is shown to be endowed with an $\mathfrak{sl}(2,\mathbb R)$-coalgebra symmetry by considering the following particular expression for  $\mathcal H^{(2)}$:
\be
\mathcal H_{\rm Zk} \equiv \mathcal H^{(2)} = J_+^{(2)} +\gamma_1 J_3^{(2)}+\gamma_2 \left(J_3^{(2)}\right)^2 .
\label{zy}
\ee
And, obviously, $ {C}^{(2)}$ (\ref{zw}) reduces  to the square of the angular momentum constant of the motion $ \mathcal C $~(\ref{zb}). The superintegrability property arises by obtaining an additional functionally independent integral $ \mathcal I$ (or $ \mathcal I'$)  (\ref{zb}) with respect to 
$\mathcal H^{(2)}$ and $ {C}^{(2)}$.

The crucial point now is that   all the above results naturally suggest   to consider the following generalization of the 
 Zernike  Hamiltonian (\ref{zy}):
 \be
\mathcal H_N =  J_+^{(2)} + \sum_{n=1}^N \gamma_n   \left(J_3^{(2)}\right)^n = \mathbf{p}^2 + \sum_{n=1}^N \gamma_n (\mathbf{q}  \boldsymbol{\cdot} \mathbf{p})^n ,
\label{zz}
\ee
where $\gamma_n$ are arbitrary parameters and  hereafter we denote  $\mathbf{p} = p_1^2 + p_2^2$ and  $\mathbf{q}  \boldsymbol{\cdot} \mathbf{p} = q_1 p_1 + q_2 p_2$. Clearly, $\mathcal H_N $ is an integrable Hamiltonian keeping  the  same constant of the motion $ \mathcal C $~(\ref{zb}) and the 
Zernike system is the particular case $\mathcal H_{\rm Zk} \equiv  \mathcal H_2$. Therefore, the open problem is to obtain the generalization of the additional integral  $ \mathcal I$ (or $ \mathcal I'$)  (\ref{zb}) thus ensuring that $\mathcal H_N$   actually determines a superintegrable system.
In the next section, we solve    this problem  presenting  the additional integrals, say   $ \mathcal I_N$ and  $ \mathcal I'_N$, for $\mathcal H_N$ which turn out to be of higher-order in the momenta.


\sect{A new class of superintegrable  momentum-dependent Hamiltonians}
\label{s3}

Our aim now is to prove that the Hamiltonian  $\mathcal H_N $ (\ref{zz})  is superintegrable for any value of the arbitrary parameters $\gamma_n$ by explicitly finding an additional constant of the motion. Hence, when both $\gamma_1$ and $\gamma_2$ are different from zero, 
$\mathcal H_N $ can be regarded as  a generalization of the classical Zernike system through superintegrable perturbations determined by the terms $\gamma_n$ with $n>2$. Firstly, we shall consider the construction of  $\mathcal H_N $ on $\mathbf E^2$  and, secondly, we shall interpret our results on $\mathbf S^2$  and $\mathbf H^2$   following Section~\ref{s21}.


\subsection{Superintegrable systems on the Euclidean plane}
\label{s31}

Let us start by introducing a set of four types of homogeneous polynomials   depending on the  two Cartesian variables $(q_1,q_2) \in \mathbb R^2$ on $\mathbf E^2$   given by 
\be
Q^{(n-j,j)}_{ab} := Q^{(n-j,j)}_{ab} (q_1,q_2), \qquad  n,j \in \mathbb N,\qquad  0\le j \le n,\qquad a,b \in \{e,o\},
\ee
where $e$ stands for \emph{even} and $o$   for \emph{odd}  according to  the parity of the integers $n$ and $j$. 
These polynomials  are 
of degree $n=(n-j) +j$ and read
\begin{equation}
\begin{split}
\label{eq:Q_ab}
Q^{(n-j,j)}_{ee} &:= (-1)^{j/2} \sum_{k=0}^{j/2} (-1)^k  \binom{n}{2k} \bigg[ (-1)^{\frac{n}{2}+1}q_1^{n-2k} q_2^{2k} + q_1^{2k} q_2^{n-2k} \bigg] , \\
Q^{(n-j,j)}_{eo} &:= (-1)^{(j-1)/2} \sum_{k=0}^{(j-1)/2} (-1)^k \binom{n}{2k+1}  \bigg[ (-1)^{\frac{n}{2}} q_1^{n-(2k+1)} q_2^{2k+1} 
  +   q_1^{2k+1} q_2^{n-(2k+1)} \bigg]  ,\\
Q^{(n-j,j)}_{oe} &:= (-1)^{j/2} \left[ \sum_{k=0}^{j/2-1} (-1)^k (-1)^{\frac{n+1}{2}} \binom{n}{2k+1} q_1^{n-(2k+1)} q_2^{2k+1}  + \sum_{k=0}^{j/2} (-1)^k \binom{n}{2k} q_1^{2k} q_2^{n-2k}
\right],
\\
Q^{(n-j,j)}_{oo} &:= (-1)^{(j-1)/2} \sum_{k=0}^{(j-1)/2} (-1)^k \bigg[ (-1)^{\frac{n+1}{2}} \binom{n}{2k} q_1^{n-2k} q_2^{2k}  + \binom{n}{2k+1} q_1^{2k+1} q_2^{n-(2k+1)} \bigg]  .\\
\end{split}
\end{equation}
Note that for some values of $j$,  some of the sums above could be empty (in particular, this may happen with $Q^{(n-j,j)}_{oe}$). 
Next   we define a set of polynomials $Q^{(n-j,j)}$ which encompasses the above four types   $Q^{(n-j,j)}_{ab}$. Let us consider the function $\Theta : \mathbb N \to \{0,1\}$,  
\begin{equation}
\Theta (m) := 
\begin{cases}
1 &\text{if } m \text{ is even} \\
0 &\text{if } m \text{ is odd} \\
\end{cases} .
\end{equation}
Then  the polynomials $Q^{(n-j,j)}$ are defined by 
\begin{equation}
\begin{split}
\label{eq:Q}
Q^{(n-j,j)} :=Q^{(n-j,j)}  (q_1,q_2)&= \Theta(n) \Theta(j) Q^{(n-j,j)}_{ee} + \Theta(n) \big(1-\Theta(j) \big) Q^{(n-j,j)}_{eo} \\
&\quad +\big(1-\Theta(n) \big) \Theta(j) Q^{(n-j,j)}_{oe} + \big(1-\Theta(n) \big) \big(1-\Theta(j) \big) Q^{(n-j,j)}_{oo} ,
\end{split}
\end{equation}
where $Q^{(n-j,j)}_{ab}$ are given by \eqref{eq:Q_ab}. Thus, the degree of the polynomial $Q^{(n-j,j)}$ is again $n$. With the previous definitions, we have all the ingredients to state and prove the main result of this paper. 
\smallskip

\begin{theorem}
\label{teor1}
Let $\{ q_1,q_2,p_1,p_2\}$ be a set of canonical Cartesian variables such that $\{ q_i,p_j\} = \delta_{ij}$. The Hamiltonian   (\ref{zz})  on the Euclidean plane, namely, 
\begin{equation}
\mathcal H_N = \mathbf{p}^2 + \sum_{n=1}^N \gamma_n (\mathbf{q}  \boldsymbol{\cdot} \mathbf{p})^n ,
\label{hamN}
\end{equation}
such that $\gamma_n$ are arbitrary parameters, is   superintegrable for all $N \in \mathbb N$. The two integrals of the motion are the usual angular momentum 
 $\mathcal C = q_1 p_2 - q_2 p_1 $  (\ref{zb}), together with the following  $N$th-order in the momenta function  
\begin{equation}
\label{eq:I}
\mathcal I_N = p_2^2 + \sum_{n=1}^N \gamma_{n} \sum_{j=0}^{\varphi(n)} p_2^{n-j} p_1^j \, Q^{(n-j,j)}(q_1,q_2),
\end{equation}
where $Q^{(n-j,j)}$ is given by \eqref{eq:Q} through \eqref{eq:Q_ab} and $\varphi(n)$ denotes the greatest even integer less than $n$, that is,
\begin{equation}
\varphi (n) := 
\begin{cases}
n-2 &\text{if } n \text{ is even} \\
n-1 &\text{if } n \text{ is odd} \\
\end{cases} .
\label{a1}
\end{equation}
 The set   $\{\mathcal H_N, \mathcal I_N, \mathcal C\}$ is formed by three functionally independent functions.
\end{theorem}

\begin{proof}
The functional independence among $\mathcal H_N$,  $\mathcal I_N$ and $ \mathcal C$ can be seen   from their explicit expressions; in fact, one can set all the parameters $\gamma_n\equiv 0$ recovering the superintegrability of the geodesic motion on the Euclidean plane.

In order to prove that $\mathcal I_N$ is an integral of motion, let us  denote 
\be
\mathcal G_N = (\mathbf{q} \boldsymbol{\cdot}  \mathbf{p})^N ,\qquad \mathcal J_N = \sum_{j=0}^{\varphi(N)} p_2^{N-j} p_1^j \,Q^{(N-j,j)},
\label{b1}
\ee
 and   we have  that
\begin{equation}
\begin{split}
\mathcal H_N &= \mathcal H_{N-1} + \gamma_N \mathcal G_N , \qquad
\mathcal I_N = \mathcal I_{N-1} + \gamma_N \mathcal J_N . \\
\end{split}
\label{xa}
\end{equation}
We also write the free motion as  $\mathcal H_0 = p_1^2 + p_2^2$ and $\mathcal I_0 = p_2^2$, and thus
\begin{equation}
\begin{split}
\mathcal H_N &= \mathcal H_{0} + \sum_{n=1}^N \gamma_n \mathcal G_n , \qquad
\mathcal I_N = \mathcal I_{0} + \sum_{n=1}^N \gamma_n \mathcal J_n . \\
\end{split}
\label{xb}
\end{equation}
From (\ref{xa}) and by bilinearity of the Poisson bracket we find that
\begin{equation}
\begin{split}
\label{eq:poisHNIN}
\{ \mathcal H_N, \mathcal I_N \} &= \{ \mathcal H_{N-1} + \gamma_N \mathcal G_N, \mathcal I_{N-1} + \gamma_N \mathcal J_N \}  \\
&=\{ \mathcal H_{N-1}, \mathcal I_{N-1}  \} + \gamma_N \left( \{ \mathcal G_{N}, \mathcal I_{N-1}  \}  + \{ \mathcal H_{N-1}, \mathcal J_{N}  \}\right) + \gamma_N^2 \{ \mathcal G_N, \mathcal J_N \}  .
\end{split}
\end{equation}
Since the only dependence on $\gamma_N$ is explicit on the previous formula, both terms $\{ \mathcal G_{N}, \mathcal I_{N-1}  \}  + \{ \mathcal H_{N-1}, \mathcal J_{N} \} $ and $\{ \mathcal G_N, \mathcal J_N \} $ must vanish. Using (\ref{xb}) and bilinearity $(N-1)$ times we obtain that
\begin{equation}
\begin{split}
\{ \mathcal G_{N}, \mathcal I_{N-1}  \} &= \{ \mathcal G_{N}, \mathcal I_{0}  \} + \sum_{n=1}^{N-1} \gamma_{n} \{ \mathcal G_{N}, \mathcal J_{n}  \} , \\
\{ \mathcal H_{N-1}, \mathcal J_{N}  \} &= \{ \mathcal H_{0}, \mathcal J_{N}  \} + \sum_{n=1}^{N-1} \gamma_{n} \{ \mathcal G_{n}, \mathcal J_{N}  \} .
\end{split}
\end{equation} 
Therefore, it remains to prove that
\begin{enumerate}

\item[1. ] $\{ \mathcal G_M, \mathcal J_N \} = 0$ for all $M,N \in \mathbb N$, and

\item[2. ] $\{ \mathcal G_{N}, \mathcal I_{0}  \} + \{ \mathcal H_{0}, \mathcal J_{N}  \} = 0$ for all $N \in \mathbb N$.

\end{enumerate}
Hence, if both of the above  two   statements are true, then from \eqref{eq:poisHNIN} and applying bilinearity $(N-1)$ times we find that
\begin{equation}
\{ \mathcal H_N, \mathcal I_N \} = \{ \mathcal H_{N-1}, \mathcal I_{N-1} \} = \cdots = \{ \mathcal H_0, \mathcal I_0 \} =\left \{ p_1^2 + p_2^2, p_2^2\right \} = 0 ,
\end{equation}
for all $N \in \mathbb N$.

Now we prove that $\{ \mathcal G_M, \mathcal J_N \} = 0$ for all $M,N \in \mathbb N$. A simple computation shows that
\begin{equation}
\begin{split}
\{ \mathcal G_M, \mathcal J_N \} &= \sum_{j=0}^{\varphi(N)} \left \{ (\mathbf{q} \boldsymbol{\cdot}  \mathbf{p})^M, \, p_2^{N-j} p_1^j \, Q^{(N-j,j)} \right \} \\
&= M (\mathbf{q} \boldsymbol{\cdot}  \mathbf{p})^{M-1} \sum_{j=0}^{\varphi(N)} p_2^{N-j} p_1^j \left( N Q^{(N-j,j)} - \mathbf{q} \boldsymbol{\cdot}  \nabla Q^{(N-j,j)} \right) = 0 ,
\end{split}
\end{equation}
where the last identity follows directly from Euler's homogeneous function theorem by recalling that $Q^{(N-j,j)}$ (\ref{eq:Q}) is a homogeneous polynomial of degree $N$, thus satisfying
\begin{equation}
\mathbf{q} \boldsymbol{\cdot}  \nabla Q^{(N-j,j)} = N Q^{(N-j,j)} ,\qquad \forall N,j \in \mathbb N .
\end{equation}
 
To prove that $\{ \mathcal G_{N}, \mathcal I_{0}  \} + \{ \mathcal H_{0}, \mathcal J_{N}  \} = 0$ for all $N \in \mathbb N$, we first compute
\begin{equation}
\label{eq:GNI0}
\{ \mathcal G_{N}, \mathcal I_{0}  \} =\left \{ (\mathbf{q}  \boldsymbol{\cdot}  \mathbf{p})^N, p_2^2 \right \} = 2 \sum_{j=0}^{N-1} (N-j) \binom{N}{j} p_2^{N-j+1} p_1^j \,q_2^{N-j-1} q_1^j ,
\end{equation}
and then
\begin{equation}
\label{eq:H0JN}
\begin{split}
\{ \mathcal H_{0}, \mathcal J_{N}  \} &= \sum_{j=0}^{\varphi(N)}  \left \{ p_1^2 + p_2^2, \, p_2^{N-j} p_1^j Q^{(N-j,j)} \right\} =   \sum_{j=0}^{\varphi(N)} p_2^{N-j} p_1^j \left \{ p_1^2 + p_2^2, Q^{(N-j,j)} \right \} \\
&= -2 \sum_{j=0}^{\varphi(N)} \left( p_2^{N-j} p_1^{j+1}\, \frac{\partial Q^{(N-j,j)}}{ \partial q_1} + p_2^{N-j+1} p_1^{j} \, \frac{\partial Q^{(N-j,j)}}{\partial q_2} \right) .
\end{split}
\end{equation}
Equating the coefficients of $p_2^{N-a} p_1^{a+1}$ in \eqref{eq:GNI0} and \eqref{eq:H0JN} we arrive at the equation
\begin{equation}
\label{xd}
\frac{\partial Q^{(N-a,a)}}{\partial q_1} + \frac{\partial Q^{(N-a-1,a+1)}}{\partial q_2} = (N-a-1) \binom{N}{a+1} q_2^{N-a-2} q_1^{a+1} .
\end{equation}
Since this computation involves the explicit expressions of $Q^{(N-j,j)}$ (\ref{eq:Q}), for the sake of brevity we only present the case when $N$ and $a$ are {\em even} numbers  ({the proof for the remaining cases is similar}). Hence $Q^{(N-a,a)}=Q^{(N-a,a)}_{ee}$ and  $Q^{(N-a-1,a+1)}=Q^{(N-a-1,a+1)}_{eo}$ given in (\ref{eq:Q_ab}), so that   we have
\begin{equation}
\label{xc}
\begin{split}
&\!\!\!\! \!\!\!\! \frac{\partial Q^{(N-a,a)}_{ee}}{\partial q_1} + \frac{\partial Q^{(N-a-1,a+1)}_{eo}}{\partial q_2}  \\
&= (-1)^{a/2} \sum_{k=0}^{a/2} (-1)^k \left[(-1)^{\frac{N}{2}}\bigg(-(N-2k) \binom{N}{2k} +(2k+1) \binom{N}{2k+1} \bigg) q_1^{N- 2k-1} q_2^{2k} \right. \\
  &\qquad\qquad\qquad+2k \left. \binom{N}{2k} q_1^{2k-1} q_2^{N-2k} + (N- 2k-1) \binom{N}{2k+1} q_1^{2k+1} q_2^{N-2k-2} \right] \\
&=  (-1)^{a/2} \sum_{k=0}^{a/2} (-1)^k \left[ 2k \binom{N}{2k} q_1^{2k-1} q_2^{N-2k} + (N- 2k-1) \binom{N}{2k+1} q_1^{2k+1} q_2^{N-2k-2} \right] \\
&=(N-a-1) \binom{N}{a+1} q_1^{a+1} q_2^{N-a-2} +  (-1)^{a/2} \left[ \sum_{k=0}^{a/2} (-1)^k \, 2k \binom{N}{2k} q_1^{2k-1} q_2^{N-2k} \right. 
 \\
  &\qquad\qquad\qquad + \left.  \sum_{k=0}^{a/2-1} (-1)^k (N-2k-1 ) \binom{N}{2k+1} q_1^{2k+1} q_2^{N-2k-2} \right]   ,
\end{split}
\end{equation}
where we have used that 
\be
 \binom{N}{b}(N-b) = \binom{N}{b+1} (b+1).
\label{xl}
\ee
 Therefore, if we prove that the last expression in (\ref{xc})  between square brackets vanishes we would have finished since the relation (\ref{xd}) would be   fulfilled. We can rewrite such  expression as
\begin{equation}
\begin{split}
\sum_{k=0}^{a/2} &(-1)^k \, 2k \binom{N}{2k} q_1^{2k-1} q_2^{N-2k} + \sum_{k=1}^{a/2} (-1)^{k-1} (N- 2k+1) \binom{N}{2k-1} q_1^{2k-1} q_2^{N-2k} \\
&= \sum_{k=1}^{a/2} (-1)^k \left [ 2k \binom{N}{2k} - (N- 2k+1) \binom{N}{2k-1} \right] q_1^{2k-1} q_2^{N-2k} = 0 ,
\end{split}
\end{equation}
where we have used again the   property (\ref{xl}).   Consequently,  we have proved that $\mathcal I_N$ (\ref{eq:I}) is an $N$th-order in the momenta integral of the motion for the Hamiltonian $\mathcal H_N$ (\ref{hamN}).
\end{proof}

By symmetry of the Hamiltonian $\mathcal H_N$ (\ref{hamN}), it is clear that the permutation of indices $1\leftrightarrow 2$    in \eqref{eq:I} provides another integral of the motion  $\mathcal I'_N$ which, obviously, is not functionally independent of the three functions presented in Theorem~\ref{teor1}: $\mathcal C$, $\mathcal I_N$ and $\mathcal H_N$. 
In addition, there exists a relationship among the above four functions. These results are characterised by the following statement.

\begin{proposition}
\label{prop1}
(i) The Hamiltonian $\mathcal H_N$ (\ref{hamN}) is also endowed with the $N$th-order in the momenta integral of motion    given by
\begin{equation}
\label{eq:II}
\mathcal I_N' = p_1^2 + \sum_{n=1}^N \gamma_{n} \sum_{j=0}^{\varphi(n)} p_1^{n-j} p_2^j \, Q^{(n-j,j)}(q_2,q_1),
\end{equation}
where $\varphi(n)$ is defined by (\ref{a1}) and $Q^{(n-j,j)}(q_2,q_1)$ are the homogeneous polynomials (\ref{eq:Q_ab}) and (\ref{eq:Q}) 
obtained through the interchange $q_1\leftrightarrow q_2$, that is, $\mathcal I_N'(q_1,p_1,q_2,p_2)=\mathcal I_N(q_2,p_2,q_1,p_1)$ (\ref{eq:I}). 
The set   $\{\mathcal H_N, \mathcal I'_N, \mathcal C\}$ is formed by three functionally independent functions.

\noindent
(ii)   The four functions  $\{\mathcal H_N, \mathcal I_N,  \mathcal I'_N, \mathcal C\}$ are subjected to the   relation
\begin{equation}
\label{a2}
\mathcal H_N=\mathcal I_N + \mathcal I_N' + \!\! \sum_{k=1}^{\varphi(N+1)/2} \!\!  (-1)^{k} \,  \gamma_{2k}\,\mathcal C^{2k}   .
\end{equation}
\end{proposition}

\begin{proof}
The only non-trivial fact  to be proved is that the relationship (\ref{a2}) holds. The procedure is   similar to the one performed in the proof of Theorem~\ref{teor1} which is quite cumbersome. We  thus restrict ourselves to outline the main steps of the proof.

Let us consider the functions $\mathcal G_N$ and $ \mathcal J_N $ (\ref{b1}) along with a new function $  \mathcal J'_N$ related to $\mathcal I_N' $ (\ref{eq:II})  in the form
\be
  \mathcal J'_N = \sum_{j=0}^{\varphi(N)} p_1^{N-j} p_2^j \,Q^{(N-j,j)}(q_2,q_1),\qquad \mathcal I'_N = \mathcal I'_{N-1} + \gamma_N \mathcal J'_N ,
  \label{a4}
\ee
that is, $ \mathcal J'_N(q_1,p_1,q_2,p_2)=\mathcal J_N(q_2,p_2,q_1,p_1)$. Next, after some long computations, we obtain the following relations according to the parity of $N$:
\begin{equation}
\label{a5}
\begin{array}{ll}
\displaystyle{N\ \mbox{\rm even:}}
   &\quad\displaystyle{  \mathcal G_N=\mathcal J_N + \mathcal J_N' +   (-1)^{N/2} \,   \mathcal C^{N}     } .   
  \\[4pt]
\displaystyle{N\ \mbox{\rm odd:}}   &\quad\displaystyle{    \mathcal G_N=\mathcal J_N + \mathcal J_N'     } .
\end{array}
\end{equation}
Now we proceed by applying mathematical induction. It is straightforward to show that the relation (\ref{a2}) holds for a low value of $N$. 
Assuming that  (\ref{a2})  is valid for $(N-1)$ we shall  prove that such equation   also holds   for $N$, distinguishing the parity of $N$.

Firstly, let $N$ be {\em even}. By taking into account (\ref{xa}), the expression (\ref{a2}) for $\mathcal H_{N-1}$ and (\ref{a5}), we obtain that
\begin{equation}
\label{a6}
\begin{split}
\mathcal H_{N}&=\mathcal H_{N-1}+\gamma_N \mathcal G_N\\
&=\mathcal I_{N-1} + \mathcal I_{N-1}' + \!\! \sum_{k=1}^{\varphi(N)/2} \!\!  (-1)^{k} \,  \gamma_{2k}\,\mathcal C^{2k}+ \gamma_N\left(\mathcal J_N + \mathcal J_N' +   (-1)^{N/2} \,   \mathcal C^{N}  \right) .
\end{split}
\end{equation}
The equations (\ref{xa}) and (\ref{a4}) lead to $\mathcal I_{N}$ and $ \mathcal I_{N}'$ in the above result. Since ${\varphi(N+1)/2}=N/2$ we also recover the complete sum in the relation (\ref{a2}) (note that ${\varphi(N)/2}=N/2-1$). And, secondly, let $N$ be {\em odd}. Now we find that 
\begin{equation}
\label{a7}
\begin{split}
\mathcal H_{N} =\mathcal I_{N-1} + \mathcal I_{N-1}' + \!\! \sum_{k=1}^{\varphi(N)/2} \!\!  (-1)^{k} \,  \gamma_{2k}\,\mathcal C^{2k}+ \gamma_N\left(\mathcal J_N + \mathcal J_N'   \right) .
\end{split}
\end{equation}
In this case, ${\varphi(N) }={\varphi(N+1) }=N-1$, so that we have proven the relation (\ref{a2}).
 \end{proof}

The results of Proposition~\ref{prop1} strongly indicate a quite different behaviour of the Hamiltonian $\mathcal H_N$ (\ref{hamN})  according to the superposition of either even or odd potential terms  $\mathcal G_n= (\mathbf{q} \boldsymbol{\cdot}  \mathbf{p})^n$ determined by the coefficients $\gamma_n$. In fact,   if we only consider {\em odd} terms in the potential, \emph{i.e.}~$\gamma_{2k} = 0$ for all $k \in \mathbb N$, then the relation  (\ref{a2})  reduces to
\begin{equation}
\mathcal H_N= \mathcal I_N + \mathcal I_N'  .
\end{equation}

Consequently,  Theorem~\ref{teor1} and Proposition~\ref{prop1} extend the results for the classical Zernike Hamiltonian of Theorem~\ref{teor0} to any arbitrary superposition of momentum dependent potentials $ (\mathbf{q} \boldsymbol{\cdot}  \mathbf{p})^n$.  In particular, setting $N=2$  we find that
\be
\mathcal H_{\rm Zk} \equiv\mathcal H_2,\qquad  \mathcal I\equiv  \mathcal I_2  ,\qquad  \mathcal I'\equiv  \mathcal I'_2  
,\qquad \mathcal H_2=\mathcal I_2 + \mathcal I_2' - \gamma_{2}\,\mathcal C^{2},
\ee
thus recovering, as a particular case of (\ref{a2})  (note that $\varphi(3)/2=1$),  the equation (75) in~\cite{PWY2017zernike}.

 In addition, it is immediate to express all the above results in  polar variables by means of the canonical transformation  (\ref{ze}); for instance, the Hamiltonian 
  $\mathcal H_N$ (\ref{hamN}) becomes
  \be
  \mathcal H_N=p_r^2 + \frac{p_\phi^2 }{r^2} +\sum_{n=1}^N \gamma_n (r p_r)^n.
  \label{a8}
 \ee


\subsection{Superintegrable systems on the   sphere and the hyperbolic space}
\label{s32}

Taking into account the interpretation carried out in Section~\ref{s21} for the Zernike system on curved spaces   and presented   in Proposition~\ref{prop0}, we  can now apply the results of Theorem~\ref{teor1} and Proposition~\ref{prop1} 
to $\mathbf S^2$ and $\mathbf H^2$. This is summarized  in the following statement.


\begin{proposition} 
\label{prop2}
Let $\{ \rho,\phi, p_\rho ,p_\phi\}$ be a set of canonical  geodesic polar  variables.  \\
(i) The   superintegrable    Hamiltonian  $\mathcal H_N$ (\ref{hamN})   can be written in these variables on $\mathbf S^2$ and $\mathbf H^2$, with constant Gaussian curvature $\kappa=-\gamma_2\ne 0$, through the 
 canonical transformation (\ref{zl}), namely
 \be
\mathcal H_N =  p_\rho^2 + \frac{p_\phi^2 }{ \Sk^2_\kk( \rho)} + \gamma_1  \Tk_\kk( \rho)\,p_\rho + \sum_{n=3}^N \gamma_n\bigl( \Tk_\kk( \rho)\,p_\rho \bigr)^n .
\label{a9}
\ee
The domain for  the variables $(\rho, \phi) $    is  again given by  $\phi\in[ 0, 2\pi)$ and (\ref{zn}).
\\
(ii) By means of the  canonical transformation (\ref{zl2}),    the    superintegrable    Hamiltonian  $\mathcal H_N$ (\ref{hamN})   can alternatively be expressed as
\bea
&& \mathcal H_N=\mathcal T_\kk+\mathcal U_\kk(\rho)+\mathcal V_\kk(\rho,p_\rho) ,\qquad \mathcal T_\kk =  p_\rho^2 + \frac{p_\phi^2 }{ \Sk^2_\kk( \rho)}  , \nonumber\\
&&\mathcal U_\kk(\rho)= -\frac{\gamma_1^2}{4}  \Tk^2_\kk( \rho)+\sum_{n=3}^N (-1)^n \gamma_n\left( \frac{\gamma_1}2\right)^{n}\!  \Tk^{2n}_\kk( \rho)
, \label{a10}\\
&&\mathcal V_\kk(\rho,p_\rho) = \sum_{n=3}^N\gamma_n  \sum_{k=1}^n (-1)^{n-k}  \binom{n}{k} \left(  \frac{\gamma_1 }2\right)^{{n-k}} \! \Tk^{2n-k}_\kk( \rho)p_\rho ^k,
\nonumber
\eea
that is,  $\mathcal T_\kk$ is the kinetic energy  (\ref{zo}) on the curved space, $\mathcal U_\kk(\rho)$ is a central   potential and $\mathcal V_\kk(\rho,p_\rho) $ is a higher-order momentum-dependent potential.  
\end{proposition}

\begin{proof}
It is a direct consequence of Theorem \ref{teor1},  the canonical transformations \eqref{zl} and  (\ref{zl2})  along with  the definitions \eqref{zi} and \eqref{zj}.
\end{proof}


The expression (\ref{a9}) shows that the initial Hamiltonian $\mathcal H_N$  (\ref{hamN})  can be seen as a    superposition of  higher-order momentum-dependent potentials (except for the quadratic term) on $\mathbf S^2$ and $\mathbf H^2$, similarly to the Euclidean case. However, 
in its alternative form (\ref{a10}), $\mathcal H_N$  can be regarded as a superposition of   the isotropic 1\,:\,1   curved oscillator  with even-order     anharmonic curved oscillators~\cite{Ballesteros2007,JNMP2008} within $\mathcal U_\kk(\rho)$ plus another superposition of  higher-order momentum-dependent potentials through the term  $\mathcal V_\kk(\rho,p_\rho) $.  
In this respect,  it is worth stressing  the prominent role played by the coefficient $\gamma_1$ in the expression  (\ref{a10}) in contrast to  (\ref{a9}). Since $\gamma_1$ is arbitrary we can set it equal to zero so that   both   expressions for  $\mathcal H_N$ (\ref{a9})  and (\ref{a10}) do coincide (and both canonical transformations (\ref{zl}) and (\ref{zl2}) as well).
In this case,  $\mathcal U_\kk(\rho)$ and $\mathcal V_\kk(\rho,p_\rho) $ (\ref{a10})  reduce to
\be
\gamma_1=0\!:\qquad \mathcal U_\kk(\rho)=0,\qquad  \mathcal V_\kk(\rho,p_\rho)=  \sum_{n=3}^N\gamma_n   \! \Tk^{n}_\kk( \rho)p_\rho ^n ,
\ee
and, consequently, there does not exist   an alternative interpretation in terms of  curved oscillators.

 We also remark that the flat limit $\kk\to 0$  ({\em i.e.}, $\gamma_2=0$) is well defined in all the results of Proposition~\ref{prop2}  leading to the corresponding expressions in $\mathbf E^2$ in a consistent way. Recall that under this  flat limit the geodesic parallel coordinates $(\rho,\phi)$ reduce to the usual polar ones $(r,\phi)$ (see (\ref{zi}) and (\ref{zj})).
 In particular, if we apply the limit $\kk\to 0$ to 
$\mathcal H_N$  (\ref{a9}) we just recover its form in polar variables (\ref{a8}) with $\gamma_2=0$. And if we now compute the limit  $\kk\to 0$ on the expressions (\ref{a10})  we directly obtain that
\bea
&& \mathcal H_N=\mathcal T_0+\mathcal U_0(r)+\mathcal V_0(r,p_r) ,\qquad \mathcal T_0 =  p_r^2 + \frac{p_\phi^2 }{  r^2}  , \nonumber\\
&&\mathcal U_0(r)= -\frac{\gamma_1^2}{4} \, r^2 +\sum_{n=3}^N (-1)^n \gamma_n \left( \frac{\gamma_1}{2} \right)^n\!  r^{2n}  
, \label{a110}\\
&&\mathcal V_0(r,p_r) = \sum_{n=3}^N\gamma_n  \sum_{k=1}^n (-1)^{n-k}  \binom{n}{k} \left(  \frac{\gamma_1}{2}    \right)^{n-k}\!  r^{2n-k} p_r ^k. 
\nonumber
\eea
The central potential $\mathcal U_0(r)$ corresponds to a superposition of anharmonic  Euclidean oscillators which, in arbitrary dimension,   were proposed in~\cite{Ballesteros2007,JNMP2008}  from a Poisson $\mathfrak{sl}(2,\mathbb R)$-coalgebra approach.
The   same results can also be obtained by applying the flat counterpart of the curved canonical transformation (\ref{zl2}) to 
 $\mathcal H_N$ (\ref{hamN}), so  with $\kk=\gamma_2=0$, or by substituting $ p_r =  \tilde p_r-  {\gamma_1}  r/2$  in $\mathcal H_N$ (\ref{a8})  with $\gamma_2=0$ and then removing the tilde in $  p_r$ (see (\ref{zzo})).


\sect{Examples and algebra of constants of the motion}
\label{s4}

In this section we illustrate the results of Theorem~\ref{teor1} and Proposition~\ref{prop1}  by explicitly writing down  the main  expressions associated with the   Hamiltonian $\mathcal H_N$  (\ref{hamN})  for some values of $N$ and, furthermore, we study the symmetry algebra, determined by the integrals of  motion, thus generalizing the cubic `Higgs' algebra (\ref{zd}) to higher-order polynomial symmetry algebras.

For this purpose we  present in Table~\ref{table1} the polynomials $Q^{(N-j,j)}$  \eqref{eq:Q}  coming from  $Q^{(N-j,j)}_{ab}$  \eqref{eq:Q_ab} which are   involved in the constant of the motion $\mathcal I_N$ (\ref{eq:I}) up to $N=8$. Thus  the expressions for 
$\mathcal I_N$ can be  obtained straightforwardly and the constant of the motion $\mathcal I'_N$ (\ref{eq:II}) can be  deduced simply by interchanging the  indices $1\leftrightarrow 2$ in the canonical variables. With this information the general relationship (\ref{a2}) among the four functions $\{H_N,\mathcal I_N,\mathcal I'_N, \mathcal C\}$, with $\mathcal C$   given by 
 (\ref{zb}), can be easily checked. These results are displayed in  Table~\ref{table2}  up to $N=6$.

As an additional  relevant property of $\mathcal H_N$  (\ref{hamN}), let us also construct its corresponding symmetry algebra understood as the algebra closed by its constants of the motion. From $\{\mathcal I_N,\mathcal I'_N, \mathcal C\}$ we define the following constants of the motion  similarly to (\ref{zzd}) (so following~\cite{PWY2017zernike}):
\be
 \mathcal L_1:= \mathcal C/2, \qquad  \mathcal L_2:=\bigl(\mathcal I'_N  -  \mathcal I_N \bigr)/2 , \qquad 
 \mathcal L_3:= \{  \mathcal L_1,\mathcal L_2\}.
\ee
Although we have not been able to deduce a general and closed expression for the  symmetry algebra  for arbitrary $N$, which remains as an open problem, we have obtained  that the three above constants of the motion satisfy the following  generic  Poisson brackets  up to $N=8$:
\be
\{ \mathcal L_1,\mathcal  L_2\}=\mathcal L_3, \qquad \{ \mathcal L_1,\mathcal  L_3\}=-\mathcal L_2, \qquad  \{ \mathcal L_2, \mathcal L_3\}=\sum_{k=0}^{N-1}  \mathcal F_k(\gamma_n,  \mathcal H_N) \mathcal L_1^{2k+1}  ,\quad\ 1\le N\le 8,
\label{a13}
\ee
where $  \mathcal F_k$ is a polynomial function depending on some coefficients belonging to the set $\{\gamma_1,\dots, \gamma_N\}$ and sometimes on  $\mathcal H_N$. Therefore, our conjecture is that for arbitrary $N$ the polynomial algebra (\ref{a13}) is of $(2N-1)$th-order and the   well-known cubic  Higgs algebra is recovered, as already shown, for the proper Zernike system with $N=2$ in (\ref{zd}).
The explicit expressions for the Poisson bracket $ \{ \mathcal L_2, \mathcal L_3\}$ are also written in  Table~\ref{table2}  up to $N=6$.


\begin{table}[t]
{\small
\caption{\small Polynomials $Q^{(N-j,j)}$  \eqref{eq:Q}  from  $Q^{(N-j,j)}_{ab}$  \eqref{eq:Q_ab}  appearing in the constant of the motion $\mathcal I_N$ (\ref{eq:I}) of  the        Hamiltonian  $\mathcal H_N$ (\ref{hamN}) up to $N=8$. All of them are homogeneous polynomials of degree $N$.}
\label{table1}
 
\noindent
\begin{tabular}{l  l}

\\[-0.2cm]

\hline

\hline

\\[-0.2cm]
$\bullet$ $N=1\quad \varphi(1)=0$: &  $Q^{(1,0)}=q_2$ \\[0.25cm]

$\bullet$ $N=2\quad \varphi(2)=0$: &   $Q^{(2,0)}=q_1^2+q_2^2$ \\[0.25cm]

$\bullet$ $N=3\quad \varphi(3)=2$: &  $Q^{(3,0)}= q_2^3$\qquad   $Q^{(2,1)}= q_1^3+ 3 q_1 q_2^2$\qquad   $Q^{(1,2)}= -q_2^3$  \\[0.25cm]

$\bullet$ $N=4\quad \varphi(4)=2$: & $Q^{(4,0)}= -q_1^4+ q_2^4$\qquad   $Q^{(3,1)}= 4\bigl( q_1^3 q_2+   q_1 q_2^3 \bigr)$\qquad   $Q^{(2,2)}=q_1^4-q_2^4$ \\[0.25cm]

$\bullet$ $N=5\quad \varphi(5)=4$: & $Q^{(5,0)}=  q_2^5$\qquad   $Q^{(4,1)}= -q_1^5 + 5 q_1 q_2^4$\qquad   $Q^{(3,2)}=5 q_1^4 q_2+ 10 q_1^2 q_2^3-q_2^5$  \\[0.20cm]
 & $Q^{(2,3)}=  q_1^5- 5 q_1 q_2^4$\qquad   $Q^{(1,4)}= q_2^5$  \\[0.25cm]

$\bullet$ $N=6\quad \varphi(6)=4$: & $Q^{(6,0)}= q_1^6 +   q_2^6$\quad  $Q^{(5,1)}=-6\bigl(   q_1^5 q_2- q_1  q_2^5\bigr)$ \ \ $Q^{(4,2)}=- q_1^6+ 15 \bigl(  q_1^4 q_2^2+ q_1^2q_2^4\bigr) - q_2^6$  \\[0.20cm]
 & $Q^{(3,3)}= 6 \bigl(  q_1^5 q_2 -   q_1q_2^5\bigr) $\qquad   $Q^{(2,4)}=  q_1^6 + q_2^6$    \\[0.25cm]

$\bullet$ $N=7\quad \varphi(7)=6$: & $Q^{(7,0)}= q_2^7$\qquad   $Q^{(6,1)}=  q_1^7+7q_1  q_2^6$ \qquad 
$Q^{(5,2)}=- 7 q_1^6 q_2 + 21 q_1^2 q_2^5- q_2^7 $ \\[0.2cm]
 & $Q^{(4,3)}=  -q_1^7 - 7 q_1 q_2^6  +  21 q_1^5q_2^2 + 35 q_1^3 q_2^4$\qquad   $Q^{(3,4)}= 7 q_1^6q_2- 21 q_1^2 q_2^5+  q_2^7$  \\[0.22cm]
  & $Q^{(2,5)}= q_1^7 + 7 q_1 q_2^6   $ \qquad  $Q^{(1,6)}=-q_2^7   $  \\[0.25cm]

$\bullet$ $N=8\quad \varphi(8)=6$: & $Q^{(8,0)}= -q_1^8 + q_2^8$\quad   $Q^{(7,1)}= 8\bigl(  q_1^7q_2+ q_1  q_2^7 \bigr)$ \quad 
$Q^{(6,2)}= q_1^8 - 28 \bigl( q_1^6q_2^2-q_1^2 q_2^6 \bigr) - q_2^8  $ \\[0.2cm]
 & $Q^{(5,3)}=  -8\bigl( q_1^7q_2+q_1 q_2^7\bigr) +56\bigl(   q_1^5 q_2^3  + q_1^3q_2^5\bigl)$\quad   $Q^{(4,4)}= -  q_1^8 +28\bigl(   q_1^6 q_2^2 - q_1^2q_2^6\bigl) + q_2^8  $  \\[0.22cm]
  & $Q^{(3,5)}= 8\bigl( q_1^7q_2  +  q_1 q_2^7\bigr)   $ \qquad  $Q^{(2,6)}=q_1^8 - q_2^8   $  \\[0.25cm]
 \hline

\hline
\end{tabular}
}
 \end{table} 



\begin{table}[htp]
{\small
\caption{\footnotesize The constant of the motion $\mathcal I_N$ (\ref{eq:I}) of  the   superintegrable     Hamiltonian  $\mathcal H_N$ (\ref{hamN}) from the polynomials written in Table~\ref{table1} up to $N=6$. Recall that   $\mathcal I'_N$ (\ref{eq:II}) comes from the interchange of indices $1\leftrightarrow 2$ in the canonical variables while $\mathcal C = q_1 p_2 - q_2 p_1 $  (\ref{zb}). The relation (\ref{a2}) among the four functions $\{H_N,\mathcal I_N,\mathcal I'_N, \mathcal C\}$ is also displayed together with the Poisson bracket $\{ \mathcal L_2, \mathcal L_3\}$ (\ref{a13})  that determines a $(2N-1)$th-order polynomial symmetry algebra.}
\label{table2}

\noindent
\begin{tabular}{l  l}
 \\[-0.2cm]
 
\hline 

\hline 

\\[-0.2cm]
$\bullet$  & \!\!\!\! \!\!\!\!   $\mathcal I_1=p_2^2+\gamma_1 Q^{(1,0)} p_2=p_2^2+\gamma_1 q_2 p_2      $\qquad $\mathcal H_1=\mathcal I_1 + \mathcal I_1' $\qquad  $    \{ \mathcal L_2, \mathcal L_3\}=  -\gamma_1^2   \mathcal L_1 $     \\[8pt]

$\bullet$   & \!\!\!\! \!\!\!\!     $\mathcal I_2=p_2^2+\gamma_1 Q^{(1,0)} p_2+  \gamma_2 Q^{(2,0)} p_2^2        =p_2^2+\gamma_1 q_2 p_2 +  \gamma_2 (q_1^2+q_2^2)p_2^2       $           \\[4pt]
 & \!\!\!\! \!\!\!\!    $ \mathcal H_2=\mathcal I_2 + \mathcal I_2' - \gamma_{2}\,\mathcal C^{2}$\qquad  $   \{ \mathcal L_2, \mathcal L_3\}=-\left( \gamma_1^2 + 2 \gamma_2 \mathcal H_2\right)\mathcal L_1-   \gamma_2^2  (2 \mathcal L_1 )^3  $  \\[8pt]
   
  $\bullet$  & \!\!\!\! \!\!\!\!     $\mathcal I_3=p_2^2+\gamma_1 Q^{(1,0)} p_2+  \gamma_2 Q^{(2,0)} p_2^2   + \gamma_3\left(Q^{(3,0)} p_2^3+Q^{(2,1)} p_2^2p_1+  Q^{(1,2)} p_2p_1^2 \right)            $           \\[4pt]
   & \!\!\!\! \!\!\!     $\quad\,=p_2^2+\gamma_1 q_2 p_2 +  \gamma_2 (q_1^2+q_2^2)p_2^2  + \gamma_3\bigl( q_2^3 p_2^3+(q_1^3+ 3 q_1 q_2^2) p_2^2p_1-q_2^3 p_2p_1^2 \bigr)         $         \\[4pt]
 & \!\!\!\! \!\!\!\!    $ \mathcal H_3=\mathcal I_3+ \mathcal I_3' - \gamma_{2}\,\mathcal C^{2}$\qquad $   \{ \mathcal L_2, \mathcal L_3\}=-\left( \gamma_1^2 + 2 \gamma_2 \mathcal H_3\right)\mathcal L_1- \left(   \gamma_2^2- 2 \gamma_1\gamma_3\right)   (2 \mathcal L_1 )^3  - \tfrac 32 \gamma_3^2   (2 \mathcal L_1 )^5$  \\[8pt]

   $\bullet$   &  \!\!\!\! \!\!\!\!    $\mathcal I_4=p_2^2+\gamma_1 Q^{(1,0)} p_2+  \gamma_2 Q^{(2,0)} p_2^2   + \gamma_3\left(Q^{(3,0)} p_2^3+Q^{(2,1)} p_2^2p_1+  Q^{(1,2)} p_2p_1^2 \right)            $           \\[4pt]
  &  \!\!\!\! \!\!\!\!    $\qquad\quad   + \gamma_4\left(Q^{(4,0)} p_2^4+Q^{(3,1)} p_2^3p_1+  Q^{(2,2)} p_2^2p_1^2 \right)            $         \\[4pt]
    & \!\!\!\! \!\!\!\!     $\quad\,=p_2^2+\gamma_1 q_2 p_2 +  \gamma_2 (q_1^2+q_2^2)p_2^2  + \gamma_3\bigl( q_2^3 p_2^3+(q_1^3+ 3 q_1 q_2^2) p_2^2p_1-q_2^3 p_2p_1^2 \bigr)         $           \\[4pt]
          & \!\!\!\! \!\!\!\!     $\qquad\quad   + \gamma_4  \bigl( ( q_2^4-q_1^4) p_2^4+4 ( q_1^3 q_2+   q_1 q_2^3  ) p_2^3p_1+  (q_1^4-q_2^4) p_2^2p_1^2 \bigr)         $       \\[4pt]
  & \!\!\!\! \!\!\!\!    $ \mathcal H_4=\mathcal I_4+ \mathcal I_4' - \gamma_{2}\,\mathcal C^{2}+\gamma_{4}\,\mathcal C^{4}$  \\[4pt]
  & \!\!\!\! \!\!\!\!    $   \{ \mathcal L_2, \mathcal L_3\}=-\big( \gamma_1^2 + 2 \gamma_2 \mathcal H_4\big)\mathcal L_1- \big(   \gamma_2^2- 2 \gamma_1\gamma_3- 2\gamma_4 \mathcal H_4\big)  (2\mathcal L_1)^3  - \tfrac 32 \big(  \gamma_3^2 -2\gamma_2\gamma_4   \big)(2 \mathcal L_1 )^5- 2\gamma_4^2(2 \mathcal L_1 )^7 $ \\[8pt]
  
   $\bullet$   & \!\!\!\! \!\!\!\!     $\mathcal I_5=p_2^2+\gamma_1 Q^{(1,0)} p_2+  \gamma_2 Q^{(2,0)} p_2^2   + \gamma_3\left(Q^{(3,0)} p_2^3+Q^{(2,1)} p_2^2p_1+  Q^{(1,2)} p_2p_1^2 \right)            $          \\[4pt]
  &  \!\!\!\! \!\!\!\!    $\qquad\quad   + \gamma_4\left(Q^{(4,0)} p_2^4+Q^{(3,1)} p_2^3p_1+  Q^{(2,2)} p_2^2p_1^2 \right)   $
 \\[4pt]
    & \!\!\!\! \!\!\!\!    $\qquad\quad  + \gamma_5\left(Q^{(5,0)} p_2^5+Q^{(4,1)} p_2^4p_1+  Q^{(3,2)} p_2^3p_1^2+  Q^{(2,3)} p_2^2p_1^3+  Q^{(1,4)} p_2 p_1^4 \right)             $            \\[4pt]
    & \!\!\!\! \!\!\!\!     $\quad\,=p_2^2+\gamma_1 q_2 p_2 +  \gamma_2 (q_1^2+q_2^2)p_2^2  + \gamma_3\bigl( q_2^3 p_2^3+(q_1^3+ 3 q_1 q_2^2) p_2^2p_1-q_2^3 p_2p_1^2 \bigr)         $           \\[4pt]
          & \!\!\!\! \!\!\!\!     $\qquad\quad   + \gamma_4  \bigl( ( q_2^4-q_1^4) p_2^4+4 ( q_1^3 q_2+   q_1 q_2^3  ) p_2^3p_1+  (q_1^4-q_2^4) p_2^2p_1^2 \bigr)         $         \\[4pt]
           & \!\!\!\! \!\!\!\!     $\qquad\quad   + \gamma_5  \bigl( q_2^5  p_2^5+ ( 5 q_1 q_2^4 -q_1^5 ) p_2^4p_1+  (5 q_1^4 q_2+ 10 q_1^2 q_2^3-q_2^5)p_2^3p_1^2+  ( q_1^5- 5 q_1 q_2^4)p_2^2p_1^3+ q_2^5p_2 p_1^4\bigr)         $           \\[4pt]
  & \!\!\!\! \!\!\!\!   $ \mathcal H_5=\mathcal I_5+ \mathcal I_5' - \gamma_{2}\,\mathcal C^{2}+\gamma_{4}\,\mathcal C^{4}$   \\[4pt]
  &  \!\!\!\! \!\!\!\!  $   \{ \mathcal L_2, \mathcal L_3\}=-\big( \gamma_1^2 + 2 \gamma_2 \mathcal H_5\big) \mathcal L_1- \big(   \gamma_2^2- 2 \gamma_1\gamma_3- 2\gamma_4 \mathcal H_5\big)  (2\mathcal L_1)^3  - \tfrac 32 \big(  \gamma_3^2 -2\gamma_2\gamma_4  +  2\gamma_1\gamma_5\big)(2 \mathcal L_1 )^5  $ \\[4pt]
&  \!\!\!\! \!\!\!\!  $\qquad\qquad\qquad - 2\big(\gamma_4^2 - 2\gamma_3\gamma_5\big)(2 \mathcal L_1 )^7  -\tfrac 52\, \gamma_5^2 (2 \mathcal L_1 )^9  $\ \\[8pt]

    $\bullet$   & \!\!\!\! \!\!\!\!     $\mathcal I_6=p_2^2+\gamma_1 Q^{(1,0)} p_2+  \gamma_2 Q^{(2,0)} p_2^2   + \gamma_3\left(Q^{(3,0)} p_2^3+Q^{(2,1)} p_2^2p_1+  Q^{(1,2)} p_2p_1^2 \right)            $           \\[4pt]
  &  \!\!\!\! \!\!\!\!    $\qquad\quad   + \gamma_4\left(Q^{(4,0)} p_2^4+Q^{(3,1)} p_2^3p_1+  Q^{(2,2)} p_2^2p_1^2 \right)   $
 \\[4pt]
    & \!\!\!\! \!\!\!\!    $\qquad\quad  + \gamma_5\left(Q^{(5,0)} p_2^5+Q^{(4,1)} p_2^4p_1+  Q^{(3,2)} p_2^3p_1^2+  Q^{(2,3)} p_2^2p_1^3+  Q^{(1,4)} p_2 p_1^4 \right)             $           \\[4pt]
        & \!\!\!\! \!\!\!\!    $\qquad\quad  + \gamma_6\left(Q^{(6,0)} p_2^6+Q^{(5,1)} p_2^5p_1+  Q^{(4,2)} p_2^4p_1^2+  Q^{(3,3)} p_2^3p_1^3+  Q^{(2,4)} p_2^2 p_1^4 \right)             $           \\[4pt]
    & \!\!\!\! \!\!\!\!     $\quad\,=p_2^2+\gamma_1 q_2 p_2 +  \gamma_2 (q_1^2+q_2^2)p_2^2  + \gamma_3\bigl( q_2^3 p_2^3+(q_1^3+ 3 q_1 q_2^2) p_2^2p_1-q_2^3 p_2p_1^2 \bigr)         $         \\[4pt]
          & \!\!\!\! \!\!\!\!     $\qquad\quad   + \gamma_4  \bigl( ( q_2^4-q_1^4) p_2^4+4 ( q_1^3 q_2+   q_1 q_2^3  ) p_2^3p_1+  (q_1^4-q_2^4) p_2^2p_1^2 \bigr)         $            \\[4pt]
           & \!\!\!\! \!\!\!\!     $\qquad\quad   + \gamma_5  \bigl( q_2^5  p_2^5+ ( 5 q_1 q_2^4 -q_1^5 ) p_2^4p_1+  (5 q_1^4 q_2+ 10 q_1^2 q_2^3-q_2^5)p_2^3p_1^2+  ( q_1^5- 5 q_1 q_2^4)p_2^2p_1^3+ q_2^5p_2 p_1^4\bigr)         $            \\[4pt]
          & \!\!\!\! \!\!\!\!     $\qquad\quad   + \gamma_6 \big( (q_1^6 +   q_2^6)p_2^6-6 (   q_1^5 q_2- q_1  q_2^5 ) p_2^5p_1+ \big (15 (  q_1^4 q_2^2+ q_1^2q_2^4) - q_1^6- q_2^6\big)p_2^4p_1^2       $           \\[4pt]
               & \!\!\!\! \!\!\!\!     $\qquad\qquad\qquad\quad   +  6  (  q_1^5 q_2 -   q_1q_2^5 ) p_2^3p_1^3+  (q_1^6 + q_2^6)p_2^2 p_1^4 \big)           $          \\[4pt]
  & \!\!\!\! \!\!\!\!   $ \mathcal H_6=\mathcal I_6+ \mathcal I_6' - \gamma_{2}\,\mathcal C^{2}+\gamma_{4}\,\mathcal C^{4}-\gamma_{6}\,\mathcal C^{6}$    \\[4pt]
  &  \!\!\!\! \!\!\!\!  $   \{ \mathcal L_2, \mathcal L_3\}=-\big( \gamma_1^2 + 2 \gamma_2 \mathcal H_6\big) \mathcal L_1- \big(   \gamma_2^2- 2 \gamma_1\gamma_3- 2\gamma_4 \mathcal H_6\big)  (2\mathcal L_1)^3  - \tfrac 32 \big(  \gamma_3^2 -2\gamma_2\gamma_4  +  2\gamma_1\gamma_5+2\gamma_6 \mathcal H_6 \big)(2 \mathcal L_1 )^5  $  \\[4pt]
&  \!\!\!\! \!\!\!\!  $\qquad\qquad\qquad - 2\big(\gamma_4^2 - 2\gamma_3\gamma_5+ 2\gamma_2\gamma_6\big)(2 \mathcal L_1 )^7  -\tfrac 52\big( \gamma_5^2-2\gamma_4\gamma_6 \big)(2 \mathcal L_1 )^9- 3 \gamma_6^2(2 \mathcal L_1 )^{11}   $ \\[8pt]

  \hline

\hline
\end{tabular}

}
 \end{table} 

 \newpage
 \clearpage
 
 \pagebreak[4]
 
\sect{Superintegrable perturbations of   the classical Zernike system}
\label{s5}

So far, we have proven the superintegrability property of the   Hamiltonian  $\mathcal H_N$  (\ref{hamN}) on $\mathbf E^2$ in Theorem~\ref{teor1} and, then,   established a natural  interpretation of these results on $\mathbf S^2$ and $\mathbf H^2$ in Proposition~\ref{prop2}. Let us now focus  on the original  classical Zernike system.

The proper    Zernike system $\mathcal H_{\rm Zk} $ (\ref{za}) arises by setting $N=2$ in $\mathcal H_N$  (\ref{hamN})  such that the coefficient $\gamma_1$ is a pure imaginary number while $\gamma_2$ is real~\cite{PWY2017zernike}. In this section let us set
\be
\gamma_1=2{\rm i}\omega, \quad \omega\in \mathbb R,\qquad \gamma_2=-\kappa, \quad \kk\in \mathbb R.
\label{a14}
\ee
Then $\mathcal H_{\rm Zk} $ (\ref{za}) becomes
\be
  \mathcal H_{\rm Zk}= \mathbf{p}^2 +2{\rm i}\omega (\mathbf{q}  \boldsymbol{\cdot} \mathbf{p})-\kk 
  (\mathbf{q}  \boldsymbol{\cdot} \mathbf{p})^2 .
  \label{a15}
  \ee
 Thus $  \mathcal H_{\rm Zk}$ can be seen   as superposition of a   linear momentum-dependent imaginary potential and a real quadratic one on $\mathbf E^2$, or as a single  linear momentum-dependent imaginary potential on $\mathbf S^2$ $(\kk>0)$  and $\mathbf H^2$  $(\kk<0)$ with kinetic energy given by $\mathbf{p}^2-\kk   (\mathbf{q}  \boldsymbol{\cdot} \mathbf{p})^2 $.
Hence on these curved spaces   $(q_1,q_2)$ can be thought as projective coordinates.
 Recall that the problem of dealing with such   imaginary potential was already analyzed and solved in~\cite{PWY2017zernike}. 
  In fact, if we apply the canonical transformation (\ref{zl2}) to (\ref{a15}) with the identification (\ref{a14}) we obtain a real Hamiltonian (\ref{zo})
  reading as
 \be
\mathcal H_{\rm Zk} =\mathcal T_\kk+\mathcal U_\kk(\rho) ,\qquad \mathcal T_\kk =  p_\rho^2 + \frac{p_\phi^2 }{ \Sk^2_\kk( \rho)}  , \qquad \mathcal U_\kk(\rho)= \omega^2  \Tk^2_\kk( \rho)  ,
\label{a16}
\ee
reproducing the isotropic 1\,:\,1   curved  (Higgs) oscillator on $\mathbf S^2$  and $\mathbf H^2$  with frequency $\omega$ as discussed after Proposition~\ref{prop0}. Since $\mathcal H_{\rm Zk} $ determines a superintegrable system, all bounded trajectories  are   periodic  and, in this case, correspond to  ellipses, that is,  to  a Lissajous  1\,:\,1 curve~\cite{PWY2017zernike,BaHeMu13,Kuruannals}. Such trajectories can be drawn directly from the expression (\ref{a16}) or by considering their real part  from (\ref{a15}).

 From this viewpoint, if we add some $\gamma_N$-potentials  with $N\ge 3$ to $ \mathcal H_{\rm Zk}$ either in the form $\mathcal H_N$  (\ref{hamN}) or in (\ref{a10}), we obtain imaginary and real superintegrable perturbations of $\mathcal H_{\rm Zk} $.
For instance, if we consider a single $\gamma_3$-potential, we find from (\ref{a15}) a cubic superintegrable perturbation   given by
\be
  \mathcal H_3=  \mathcal H_{\rm Zk}+ \gamma_3  (\mathbf{q}  \boldsymbol{\cdot} \mathbf{p})^3,
  \label{a17}
\ee
while from (\ref{a16}) adopts the following more cumbersome expression 
\be
  \mathcal H_3=  \mathcal H_{\rm Zk}+{\rm i}\gamma_3\, \omega^3 \Tk^6_\kk( \rho)-\gamma_3\big(3 \omega^2  \Tk^5_\kk( \rho)p_\rho+ 3{\rm i}\omega     \Tk^4_\kk( \rho)p_\rho^2-   \Tk^3_\kk( \rho)p_\rho^3  \big).
  \label{a18}
\ee
The central potential determined by $\!\Tk^6_\kk( \rho)$ is real whenever $\gamma_3$ is a pure imaginary number. 
In this case, if one compute the real part of the trajectories either from (\ref{a17}) or from (\ref{a18}), one finds bounded trajectories which `deform' the ellipses associated with the initial Zernike system. Some of them are drawn in Fig.~1 with $\kk=+1$   for some imaginary values of $\gamma_3$ in the projective plane $(q_1,q_2)$ (so on the sphere). Similar trajectories arises for $\kk=-1$ (thus on $\mathbf H^2$).
 
 Likewise, we can consider a   quartic perturbation with $\gamma_3=0$ and $\gamma_4\ne 0$, that is, 
 \be
  \mathcal H_4=  \mathcal H_{\rm Zk}+ \gamma_4  (\mathbf{q}  \boldsymbol{\cdot} \mathbf{p})^4,
  \label{a19}
\ee
which in geodesic polar variables turns out to be
 \be
  \mathcal H_4=  \mathcal H_{\rm Zk}+ \gamma_4\, \omega^4 \Tk^8_\kk( \rho)+\gamma_4\big(4 {\rm i}\omega^3  \Tk^7_\kk( \rho)p_\rho-6 \omega^2  \Tk^6_\kk( \rho)p_\rho^2-4{\rm i}\omega     \Tk^5_\kk( \rho)p_\rho^3+  \Tk^4_\kk( \rho)p_\rho^4  \big).
  \label{a20}
\ee
 Then the central potential  associated with $\Tk^8_\kk( \rho)$ is real if $\gamma_4\in\mathbb R$. The  real part of the corresponding trajectories with $\kk=+1$  are shown in Fig.~2   for some real values of $\gamma_4$ in the projective plane $(q_1,q_2)$.

 \bigskip
 


\begin{figure}[H]
\label{fig1}
\begin{center}
\includegraphics[scale=0.72]{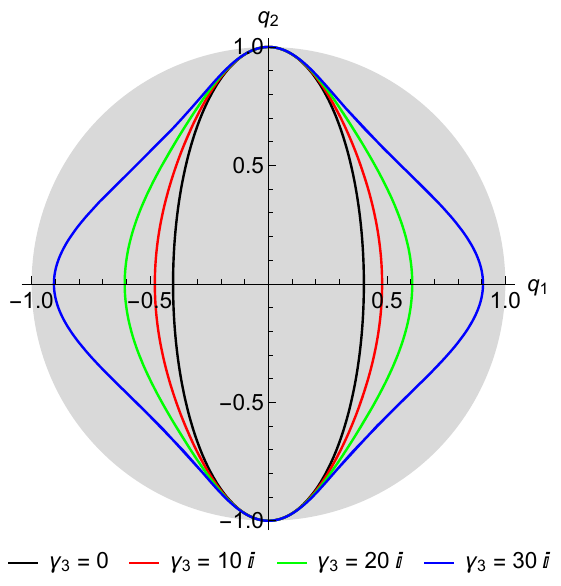} \hspace{1cm}
\includegraphics[scale=0.72]{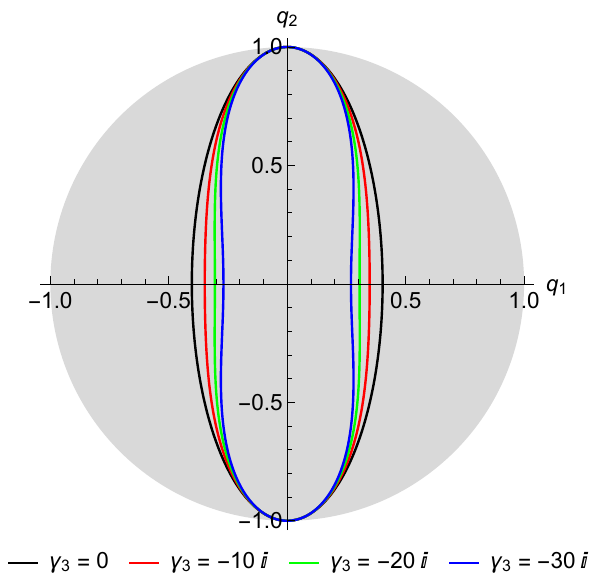} 
\caption{\small Plots of the real part of   trajectories from the cubic perturbation of the  Zernike system $ \mathcal H_3$  (\ref{a17}) with $\kk=+1$.}
\end{center}
\end{figure}

\vskip-0.25cm


 \begin{figure}[H]
 \label{fig2}
\begin{center}
\includegraphics[scale=0.72]{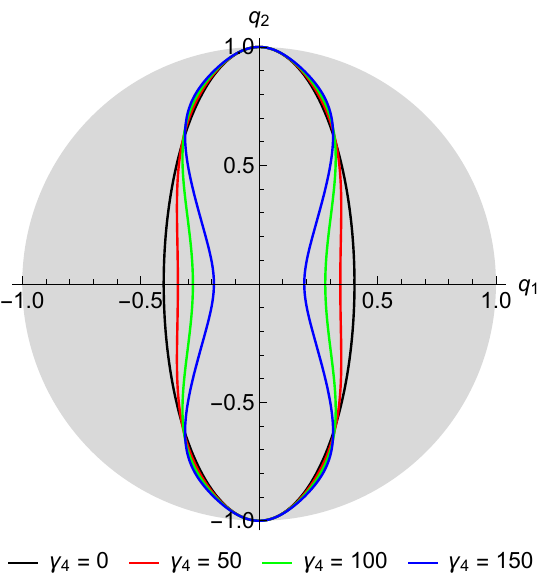} \hspace{1cm}
\includegraphics[scale=0.72]{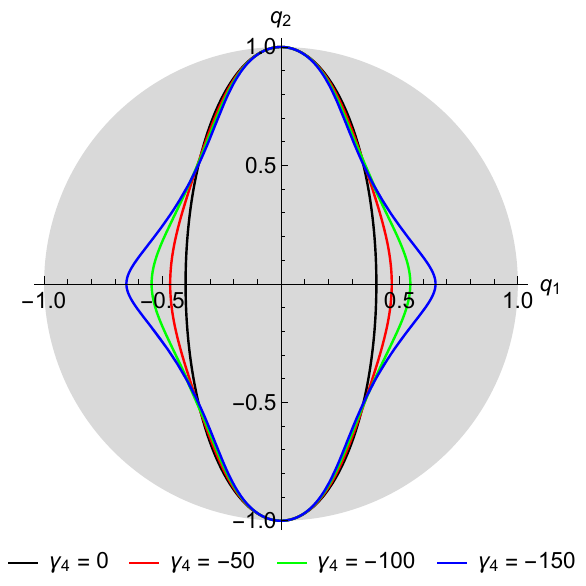} \\
\caption{\small Plots of the real part of   trajectories from the quartic perturbation of the  Zernike system  $\mathcal H_4$ (\ref{a19}) with $\kk=+1$.}
\end{center}
\end{figure}


From the expression (\ref{a10}) one can   easily check that central potential $\mathcal U_\kk(\rho)$ with 
$\gamma_1$ given by  (\ref{a14}) and with a single parameter $\gamma_N\ne 0$ ($N\ge 3$) is a real potential according to the parity of $N$: $\gamma_N$ must be a pure imaginary number when $N$ is odd, while $\gamma_N\in \mathbb R$ when $N$ is even. For these cases, it can be obtained  that the real part of the trajectory is bounded. We illustrate this fact by drawing the  
fifth-order perturbation of the   Zernike system in Fig.~3 and the sixth-order perturbation  in Fig.~4 with $\kk=+1$ and again in the projective plane $(q_1,q_2)$.

\bigskip

\begin{figure}[H]
 \label{fig3}
\begin{center}
\includegraphics[scale=0.71]{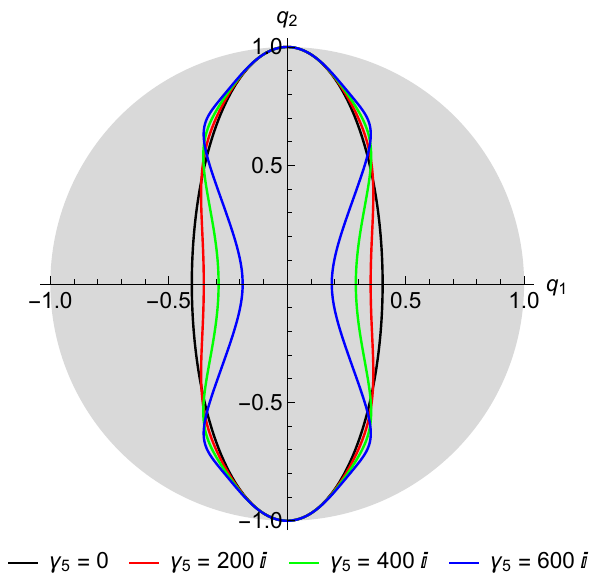} \hspace{1cm}
\includegraphics[scale=0.71]{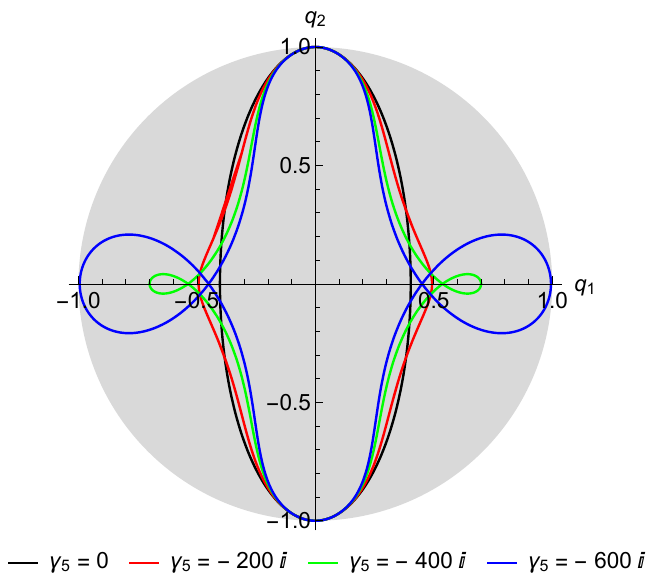} \\
\caption{\small Plots of the real part of   trajectories from the fifth-order perturbation   of the  Zernike system obtained from $\mathcal H_5$  (\ref{hamN})  under the identification (\ref{a14})  with $\kk=+1$ and with a single term $\gamma_5\ne0$.}
\end{center}
\end{figure}

\vskip-0.25cm


\begin{figure}[H]
\label{fig4}
\begin{center}
\includegraphics[scale=0.72]{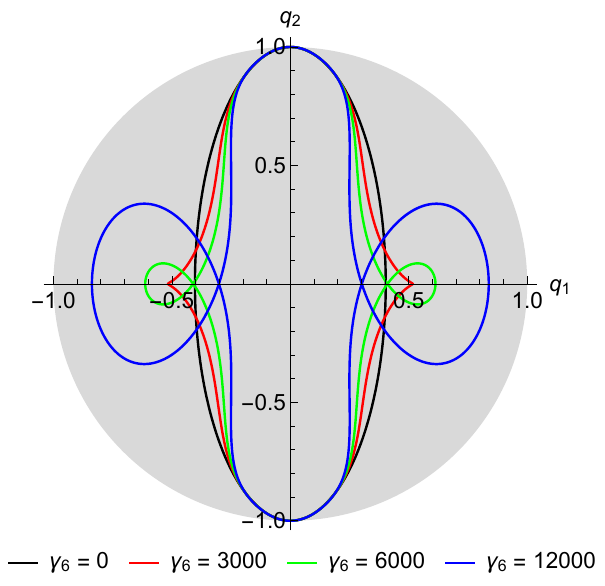} \hspace{1cm}
\includegraphics[scale=0.72]{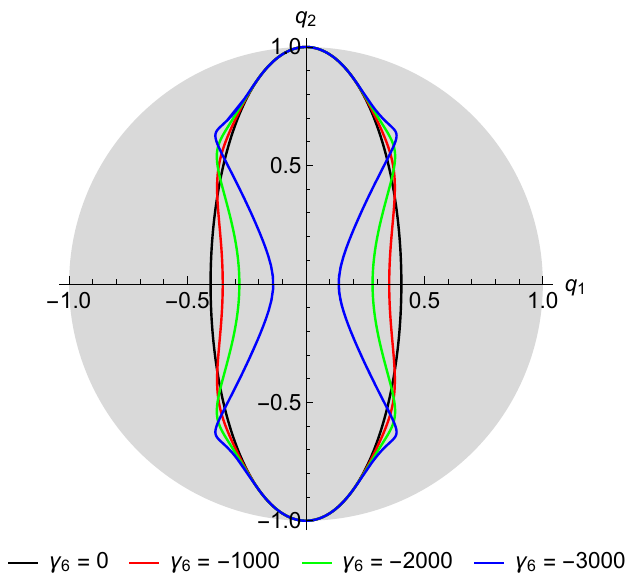} \\
\caption{\small Plots of the real part of   trajectories from the sixth-order perturbation   of the  Zernike system obtained from $\mathcal H_6$  (\ref{hamN})  under the identification (\ref{a14})  with $\kk=+1$ and with a single term $\gamma_6\ne0$.}
\end{center}
\end{figure}


 The superintegrable properties for the four particular perturbations of the   Zernike system here considered can  be extracted  straightforwardly from the general results presented in Table~\ref{table2} since this covers the cases with  $N\le 6$. Obviously, 
 one can always construct  superpositions of different higher-order  perturbations of the   Zernike system.


\sect{Conclusions and outlook}

Throughout this work we have constructed a new class of higher-order superintegrable momentum-dependent  Hamiltonians summarized in Theorem~\ref{teor1}, which allows for an arbitrary superposition of potentials   beyond the  linear and quadratic momentum-dependent  ones. Moreover, these systems  have not only been    interpreted on the 2D Euclidean plane  $\mathbf E^2$ but also on the sphere $\mathbf S^2$  and the hyperbolic plane $\mathbf H^2$ in Proposition~\ref{prop2}. The corresponding  higher-order  momentum-dependent   constants of the motion have been explicitly written and some algebraic properties have also been studied, such as the relationship among the constants of the motion  in Proposition~\ref{prop1} and the symmetry algebra of the integrals in Section~\ref{s4}.

It is worth recalling that the cornerstone of our construction is based in the superintegrable classical Zernike system~\cite{PWY2017zernike}  described in Theorem~\ref{teor0} together with its underlying Poisson $\mathfrak{sl}(2,\mathbb R)$-coalgebra symmetry, presented in Section~\ref{s22}, which holds for $\mathcal H_N$ (\ref{zz}) for any $N$.   From the latter property,    four open problems   naturally  arise which could be faced in order to generalize $\mathcal H_N$ and the results  of Theorem~\ref{teor1}:
\begin{itemize}

\item If we consider arbitrary real parameters $\lambda_i$ $(i=1,2)$ in the symplectic realization (\ref{zv}), we   obtain a new
integrable Hamiltonian $\mathcal H_{\lambda,N}$ generalizing the superintegrable $\mathcal H_N$ (\ref{zz})  via a superposition with a potential $\mathcal W_{\lambda}(q_1,q_2)$ as
\be
\mathcal H_{\lambda,N}=\mathcal H_N+\mathcal W_{\lambda}(q_1,q_2)=\mathbf{p}^2 + \sum_{n=1}^N \gamma_n (\mathbf{q}  \boldsymbol{\cdot} \mathbf{p})^n+\frac {\otra_1}{ q_1^2}  +\frac {\otra_2}{ q_2^2}\,  ,
\label{y1}
\ee
which is always endowed with the constant of the motion given by $ {C}^{(2)}$ (\ref{zw}). In $\mathbf E^2$, with $(q_1,q_2)$ identified with Cartesian coordinates, the $\lambda_i$-terms are  `centrifugal' (or Rosochatius--Winternitz) potentials such that they 
  provide centrifugal barriers when both constants are positive so restricting the trajectories to some quadrants in the Euclidean plane. In geodesic polar variables (\ref{zl2})  the additional potential $\mathcal W_{\lambda}$ becomes
  \be
\mathcal W_{\lambda}(\rho,\phi)=\frac {\otra_1}{ \Sk^2_\kk( \rho)\cos^2\phi}  +\frac {\otra_2}{ \Sk^2_\kk( \rho) \sin^2\phi}\, ,
  \ee
which can be interpreted as two {\em noncentral}    1\,:\,1 isotropic  curved oscillators on  $\mathbf S^2$ or as centrifugal barriers   on  $\mathbf H^2$ when both $\lambda_i>0$~\cite{BaHeMu13}.

\item The   $\mathfrak{sl}(2,\mathbb R)$-coalgebra symmetry~\cite{Ballesteros2007,BBHMR2009}  directly leads  to the   following quasi-maximally superintegrable generalization of the Hamiltonian  (\ref{zz}) in arbitrary dimension $d$:
\be
\mathcal H^{(d)}_{\lambda,N}=J_+^{(d)} + \sum_{n=1}^N \gamma_n   \left(J_3^{(d)}\right)^n =\sum_{i=1}^d p_i^2 + \sum_{n=1}^N \gamma_n\left( \sum_{i=1}^d q_i   p_i\right)^n+\sum_{i=1}^d \frac {\otra_i}{ q_i^2}    \, ,\qquad d\ge 2,
\label{y2}
\ee
 which, by construction, is endowed with $(2d-3)$ functionally independent `universal' constants of the motion~\cite{Ballesteros2007,BBHMR2009,Latini2019,Latini2021} and     can  be further interpreted on either $\mathbf E^d$,   $\mathbf S^d$ or $\mathbf H^d$.

\item The Hamiltonian   $\mathcal H^{(d)}_{\lambda,N}$ (\ref{y2})    can also be generalized   to    spaces of nonconstant curvature through (non-deformed) Poisson
$\mathfrak{sl}(2,\mathbb R)$-coalgebra spaces following~\cite{AEHR2007} which would allow   several possibilities for a generalized momentum-dependent potential.

\item And, finally,   the last possible generalization is   to consider 
Poisson--Hopf algebra deformations of $\mathfrak{sl}(2,\mathbb R)$~\cite{BBHMR2009,AHR2005,RBHM2007} which convey an additional quantum deformation parameter $q={\rm e}^z$ giving  rise to a  deformed classical Hamiltonian $\mathcal H^{(d)}_{z,\lambda,N}$ such that      $\lim_{z\to 0} \mathcal H^{(d)}_{z,\lambda,N}= \mathcal H^{(d)}_{\lambda,N}$. In this case, the deformation parameter $z$  would determine superintegrable perturbations of the initial (underformed) Hamiltonian  (\ref{y2}).

\end{itemize}

The crucial point to solve any of the  above four problems is to obtain the corresponding generalized counterpart of the constant of the motion $ \mathcal I_N$ (\ref{eq:I}), since both the  coalgebra and deformed coalgebra symmetries ensure the existence of $(2d-3)$ functionally independent constants of the motion. Clearly, these tasks are by no means trivial.

In contrast to the previous (open) discussion, it might be straightforward to apply  the results  of Theorem~\ref{teor1} and  Proposition~\ref{prop2} to the three  (1+1)D Lorentzian spacetimes of constant curvature, {\em i.e.},   the Minkowskian and (anti-)de Sitter spacetimes. The procedure requires to incorporate a second `contraction' parameter, say $\kk_2$, beyond the  curvature  of the space $\kk\equiv \kk_1$, depending on the speed of light $c$ as $\kk_2=-1/c^2$~\cite{conf,trigo}, which could be performed by analytic continuation. Therefore,  the `additional' constant of the motion $ \mathcal I_N$ (\ref{eq:I}) would   formally  hold but now in a 
Riemannian--Lorentzian form, so that no further cumbersome computations would be needed. For instance, under this approach the 
   Zernike  system    written as the natural Hamiltonian  given in Proposition~\ref{prop0} in geodesic polar variables (\ref{zo}),  with $\gamma_1=2 {\rm i} \omega$, turns out to be 
\be
\mathcal H_{{\rm Zk},\kk_1,\kk_2} =\mathcal T_{\kk_1,\kk_2}+\mathcal U_{\kk_1}(\rho) ,\qquad \mathcal T_{\kk_1,\kk_2} =  p_\rho^2 + \frac{p_\phi^2 }{\kk_2\Sk^2_{\kk_1}( \rho)}  , \qquad \mathcal U_{\kk_1}(\rho)=\omega^2  \Tk^2_{\kk_1}( \rho)  ,
\label{y4}
\ee
where $\mathcal T_{\kk_1,\kk_2}$ is the kinetic energy on the curved space and $\mathcal U_{\kk_1}(\rho)$ is the 1\,:\,1 isotropic curved oscillator. 
Hence, for $\kk_2=+1$ $(c={\rm i} )$  the results here presented for the three Riemannian spaces  of constant curvature would be recovered, meanwhile for $\kk_2<0$  ($c$ finite),  new results concerning Lorentzian spacetimes   would be obtained. We recall that the Hamiltonian (\ref{y4})  has been deeply studied in~\cite{HB2006} in      (2+1)-dimensions (see also~\cite{Petrosian1} for the specific  anti-de Sitter case).

To conclude, we would like to  comment   on what, in our opinion, is the main open problem of this work, which is precisely to obtain the quantum analogue of the superintegrable classical Hamiltonian $\mathcal H_N$ (\ref{zz}). Let us consider the usual quantum position $\hat{\mathbf{q}}$ and momenta  $\hat{\mathbf{p}}$ operators, with canonical Lie brackets and differential representation given by
\be
[\hat q_i,\hat p_j]={\rm i}\hbar \delta_{ij},\qquad \hat q_i \psi( \mathbf{q} ) =q_i\psi( \mathbf{q} ),\qquad \hat p_i \psi( \mathbf{q} )=-{\rm i}\hbar \,\frac{\partial   \psi( \mathbf{q} )}{\partial q_i}\, .
\label{y5}
\ee
From them, we  quantize the two-particle symplectic realization (\ref{zv}) (with $\lambda_i=0)$ in the form
\be
\hat J_-^{(2)}=\hat q_1^2+\hat  q_2^2\equiv \hat{\mathbf{q}}^2,    \qquad    
\hat J_+^{(2)}= \hat p_1^2  +\hat p_2^2\equiv \hat{\mathbf{p}}^2   ,   \qquad 
\hat J_3^{(2)}=\hat q_1 \hat p_1    + \hat q_2 \hat p_2\equiv  \hat{\mathbf{q}}  \boldsymbol{\cdot} \hat{\mathbf{p}} \,  .
\label{y6}
\ee
These operators close on a Lie algebra isomorphic to $\mathfrak{gl}(2)$:
\be
\bigl[ \hat J_3^{(2)}, \hat J_\pm^{(2)} \bigr]=\pm 2{\rm i}\hbar  \hat  J_\pm^{(2)},\qquad \bigl[ \hat J_-^{(2)}, \hat J_+^{(2)} \bigr]= 4 {\rm i}\hbar \hat J_3^{(2)}+ 4 \hbar^2 {\rm Id},
\label{y7}
\ee
where ${\rm Id}$ is the identity operator. Then we propose that the   quantization of $\mathcal H_N$ (\ref{zz}) is   defined by  the following quantum Hamiltonian
\be
\hat{\mathcal H}_N  = \hat J_+^{(2)} + \sum_{n=1}^N \gamma_n   \left(\hat J_3^{(2)}\right)^n  = \hat{\mathbf{p}}^2 + \sum_{n=1}^N \gamma_n ( \hat{\mathbf{q}}  \boldsymbol{\cdot} \hat{\mathbf{p}})^n   ,
\label{y8}
\ee
that is,
 \be
\hat{\mathcal H}_N \psi( \mathbf{q} )  =  -\hbar^2\left(\frac{\partial^2\psi( \mathbf{q} ) }{\partial q_1^2}+\frac{\partial^2\psi( \mathbf{q} ) }{\partial q_2^2} \right)+ \sum_{n=1}^N \gamma_n(-{\rm i}\hbar )^n
\left( q_1 \frac{\partial  }{\partial q_1}+   q_2 \frac{\partial  }{\partial q_2~} \right)^n\psi( \mathbf{q} )\, .
\label{y9} 
\ee
Thus $\hat{\mathcal H}_N$    is now endowed with a Lie $\mathfrak{gl}(2)$-coalgebra symmetry (instead of   a Poisson $\mathfrak{sl}(2,\mathbb R)$-coalgebra one). We stress that such a `direct' quantization does not work on the  constant of the motion $ \mathcal I_N$ (\ref{eq:I}) since serious ordering problems arise, so that additional terms  must be added in order  to obtain the quantum analogue  of $ \mathcal I_N$  and thus proving  quantum superintegrability of (\ref{y8}).

 Work on the above research  lines is currently in progress.


\section*{Acknowledgements}

\phantomsection
\addcontentsline{toc}{section}{Acknowledgements}

{This work has been partially supported by   Agencia Estatal de Investigaci\'on (Spain)  under grant  PID2019-106802GB-I00/AEI/10.13039/501100011033.}



\end{document}